\documentclass[12pt, draftclsnofoot, onecolumn]{IEEEtran}
 \usepackage{amsmath,amssymb}
 \usepackage{subfigure}
 \usepackage{graphicx,graphics,color,psfrag}
 \usepackage{cite,balance}
 \usepackage{caption}
 \allowdisplaybreaks
 \usepackage{algorithm}
 \usepackage{accents}
 \usepackage{amsthm}
 \usepackage{bm}
 \usepackage{algorithmic}
 \usepackage[english]{babel}
 \usepackage{multirow}
 \usepackage{enumerate}
 \usepackage{cases}
 \usepackage{stfloats}
 \usepackage{dsfont}
 \usepackage{color,soul}
 \usepackage{amsfonts}
 \usepackage{cite,graphicx,amsmath,amssymb}
 \usepackage{subfigure}
 \usepackage{fancyhdr}
 \usepackage{hhline}
 \usepackage{graphicx,graphics}
 \usepackage{array,color}
 \usepackage{amsmath}
 \usepackage{booktabs}

\newtheorem{theorem}{Theorem}

\newtheorem{lemma}{Lemma}

\newtheorem{proposition}{Proposition}

\newtheorem{remark}{\bf Remark}
\def\phi{\varphi}

\def\l{\left}
\def\r{\right}
\def\({\left(}
\def\){\right)}

\def\b0{{\mathbf{0}}}

\newcommand{\var}{\mathsf{var}}

\begin{document}

\title{\huge Effects of Base-Station Spatial Interdependence on Interference Correlation and Network Performance}

\author{Juan~Wen,~Min~Sheng,~\IEEEmembership{Senior~Member,~IEEE,} Kaibin Huang,~\IEEEmembership{Senior~Member,~IEEE,} Jiandong~Li,~\IEEEmembership{Senior~Member,~IEEE}
\thanks{J. Wen is with the ISN State Key Lab,  Xidian University, Xi'an, Shaanxi, China and the Department of Electrical and Electronic Engineering, the University of Hong Kong, Hong Kong (Email: juanwen66@gmail.com).
}%
\thanks{M. Sheng and J. Li are with the ISN State Key Lab, Xidian University, Xi'an, Shaanxi, China (Email: \{msheng, jdli\}@mail.xidian.edu.cn).%
}%
\thanks{K. Huang is with the Department of Electrical and Electronic Engineering, the University of Hong Kong, Hong Kong (Email: huangkb@eee.hku.hk). %
}
}
\markboth{}{Your Name \MakeLowercase{\emph{et al.}}: Your Title}
\maketitle

\begin{abstract}
The spatial-and-temporal correlation of interference has been well studied in Poisson networks where the interfering base stations (BSs) are independent of each other. However, there exists spatial interdependence including attraction and repulsion among the BSs in practical wireless networks, affecting the interference distribution and hence the network performance. In view of this, by modeling the network as a Poisson clustered process, we quantify the effects of spatial interdependence among BSs on the interference correlation and analytically prove that BS clustering increases the level of interference correlation. In particular, it is shown that the level increases as the attraction between the BSs increases. Furthermore, we study the effects of spatial interdependence among BSs on network performance with a retransmission scheme via considering heterogeneous cellular networks in which small-cell BSs exhibit a clustered topology in practice. It is shown that the interference correlation degrades the network performance and the degradation increases as the attraction between BSs increases. Finally, a correlation-aware retransmission scheme is proposed to improve the network performance by taking advantage of the interference correlation and avoiding the blind retransmissions. \end{abstract}


\newpage
\section{Introduction}
\IEEEPARstart{S}{mall}-cell BSs (SBSs) in heterogeneous cellular networks (HCNs) are deployed based on the spatial distribution of users to improve quality-of-service. To be specific, the SBSs are clustered at  hotspots where data traffic is concentrated and the clustering phenomenon is referred to as the \emph{intra-tier dependence}. On the other hand, to avoid causing  strong  inter-tier interference, SBSs are allocated sufficiently far away from  Macro-cell BSs (MBSs) and the resultant repulsion  between the SBSs and MBSs is called \emph{inter-tier dependence}. Such spatial interdependence including intra- and inter-tier dependence in HCNs significantly affects the interference correlation and hence the network performance. Nevertheless, these effects have not been quantified in the literature as the analysis is challenging. For mathematical tractability, the BSs in HCNs are commonly modeled as a multi-tier independent Poisson network where the nodes are mutually independent \cite{HCN-K-Tier,HCN-Model}. Although this model provides tractability and useful design insight, it fails to account for the spatial interdependence in BSs. To overcome this drawback, we instead model HCNs using spatial clustered processes to characterize the effects of BS spatial interdependence on interference correlation and network performance.

\subsection{Related Work}

Extensive research has been conducted on analyzing the performance of HCNs using the tool of stochastic geometry based on the most popular model of  multi-tier independent Poisson network \cite{Stochastic-Geometry-HCN-Tutorial,StoGeo}. In this model, the BSs in each tier are distributed as a Poisson point process (PPP) and tiers are independent and have different densities, transmission powers, and requirements on signal-to-interference-plus-noise ratios (SINRs). Such a model is deployed in \cite{HCN-K-Tier} to investigate the outage probability and average rate for HCNs under a SINR constraint. Similar  approaches have been adopted in extensive work  on studying the HCN performance under various network operations and designs including cell association \cite{FlexibleCellAssociation}, resource management \cite{StructuredSpectrumAllocation,HCN-FFR}, traffic offloading \cite{FlexibleCellAssociation,Tony_HCN_Offloading}, D2D communications \cite{HongGuang_15TWC,Tony_D2D_HCN,SpecSharing_Cellular_Adhoc}, energy efficient transmissions \cite{Tony_EE_HCN} and BS cooperation \cite{HCN-SpatiotemporalCooperation,InterfCorrela-ICIC,MulticellCoop_KB}. Although the PPP models capture the irregular topologies  of HCNs, they overlook a key feature of HCN, namely the BS spatial interdependence. 

In the area of stochastic geometry, there exists a rich family of spatial point processes which are suitable candidates for modeling the BS spatial interdependence in HCNs \cite{HCN_ApproxSIRAnalysis,InterferCor_NonPoisson_GC}. On one hand, Poisson hole process \cite{HCN-Dependence}, determinantal point process \cite{DPP1,DPP2}, and Ginibre point process \cite{GPP1} feature repulsion between points that can be deployed to model inter-tier BS repulsion in  HCNs. The limited analytical tractability of these point processes results in complex  network performance analysis with little simple insight   [18-21]. On the other hand, intra-tier SBS clustering in HCNs can be modeled naturally using various tractable cluster point processes, such as the Poisson cluster processes grouping Matern cluster processes (MCP) and Thomas cluster processes \cite{HCN_PoissonCluster,HCN-Dependence,PoissonClusterProcess_TIT}. In addition, a HCN model based on the second-order cluster processes (SOCP) captures both the inter-tier and intra-tier interdependence \cite{HCN_SOCP}. While the effects of BS clustering on interference distributions have been extensively studied for different types of networks (see e.g.,  \cite{HCN-Dependence,DPP1,DPP2,GPP1,HCN_PoissonCluster,PoissonClusterProcess_TIT,HCN_SOCP})", there exist few results on its effects on interference correlation (in both the space and time). It is important to note that the two types of results are different with the former concerning interference measured at a single location in plane but the latter relating interference measured at two separate locations or two separate time instants. This work makes contributions by deriving the latter results.

In wireless networks, spatial-and-temporal interference correlation arises the random spatial distribution of interfering BSs and the channel time-variations \cite{InterfCorreThreeSources}. Ganti and Haenggi are among the first to quantify the interference correlation in terms of correlation coefficient \cite{InterferenceCorreLetter}. It was discovered that such correlation reduces the diversity gain in  retransmission  and thereby degrades the network performance \cite{IntefCorrDiversityPolynomials,InterfCorreDiversityLoss,CorrelationMobileRandomNet,LocalDelay_MAC}. In particular, the performance gain of \emph{hybrid automatic repeat request} is marginalized due to the correlation \cite{HARQ-CorrInterf-adhoc}. The negative effects of interference correlation may be exacerbated by the BS spatial interdependence. This is an important issue due to the popularity of HCNs but has not yet been investigated in prior work.

\subsection{Contributions and Organization}

First, we investigate the effects of interferer's interdependence on the interference correlation.  Consider interference powers measure at two separate locations in the presence of an interferer field following one of three possible distributions, namely PPP, MCP and SCOP, where the conventional case of PPP serves as a benchmark.  To facilitate the summary of results,  let $\zeta_{\textsf P}$, $\zeta_{\textsf M}$, and $\zeta_{\textsf S}$ denote the (spatial-and-temporal) interference correlation coefficients corresponding to the PPP, MCP and SCOP, respectively.  The mean number of points and the cluster radius in the MCP and SOCP models  are represented as $\{c_{\textsf M}, R_{\textsf M}\}$ and $\{c_{\textsf S}, R_{\textsf S}\}$, respectively. Our key findings are summarized as follows.

\begin{enumerate}
\item We derive the interference-correlation coefficients $\zeta_{\textsf M}$ and $\zeta_{\textsf S}$, and show that they are greater than $\zeta_{\textsf P}$,
given identical densities. This analytically shows that  the interferer clustering increases the level of interference  correlation. Furthermore, $\zeta_{\textsf M}$
and  $\zeta_{\textsf S}$ are equal if the two corresponding models have the same cluster radii and mean numbers of points per cluster ($c_{\textsf M} = c_{\textsf S}$ and $R_{\textsf M} = R_{\textsf S}$).

\item It is shown that  the correlation coefficient $\zeta_{\textsf M}$ (or $\zeta_{\textsf S}$) is a monotone-increasing function of $c_{\textsf M}$ (or $c_{\textsf S}$) and a monotone-decreasing function of $R_{\textsf M}$ (or $R_{\textsf S}$).  In addition, $\zeta_{\textsf M}$ and $\zeta_{\textsf S}$ converge to  $\zeta_{\textsf P}$ as  $\frac{c_{\textsf M}}{R_{\textsf M}^2}$ and $\frac{c_{\textsf S}}{R_{\textsf S}^2}$ varnish.

\end{enumerate}

Next, we analyze the effect of BS interdependence on the network performance.  To this end, we consider two scenarios of downlink HCNs with different spatial  interdependence between BSs, represented by two models where MBSs are distributed as a PPP for both models while SBSs as a MCP in one model, called the \emph{MCP model}, and as a SOCP in the other, called the \emph{SOCP model}. The MCP model captures only the intra-tier (SBSs) interdependence while the SOCP reflects both the intra- and inter-tier interdependence. Moreover, HARQ is used to enhance transmission reliability.

Based on the network models, we derive the numerically integrable expressions and their bounds for the joint success probabilities, defined as the success probability in multiple successive transmissions, for macro-cell users (MUs) and small-cell users (SUs). It is found that the joint success probability for MUs in the SOCP model is larger than  that in  the MCP model. This suggests that the inter-tier interdependence enhances the MU performance. In addition, it is found that, interference correlation degrades the network performance and the degradation increases as the attraction between the BSs increases. Further, a correlation-aware retransmission scheme is proposed to improve the network performance via taking good advantage of interference correlation and effectively avoiding the blind retransmissions.

The remainder of the paper is organized as follows. The network models and metrics are described in Section II. The interference correlation and HCN performance are analyzed in Section III and IV, respectively. Numerical results are provided in Section V followed by conclusions in Section VI.

\section{Network Models and Metrics}\label{Section:System}

The network models and metrics are introduced in this section. The symbols used therein and their meanings are tabulated in Table~\ref{tab:table1}.

\begin{table}[t!]
  \centering
  \caption{Summary of Notations}
  \label{tab:table1}

  \begin{tabular}{cl}

    \toprule
    Symbol & Meaning \\
    \midrule

    $\Phi_m$, $\Phi_s$ & Point process of (MBSs,  SBSs)\\
    $\lambda_m$, $\lambda_s$ &  Density of (MBSs, SBSs)\\
    $P_m$, $P_s$ & Transmission power of (MBSs, SBSs)\\
    $\beta_m$, $\beta_s$ & SIR threshold of (MBSs, SBSs)\\
    $D_m$, $D_s$ & Coverage radius of (MBSs, SBSs)\\
    $h$  & Rayleigh fading gain with unit mean\\
    $g(x)$, $\alpha$ & Path-loss function, path-loss exponent\\
    $\Phi_{\textsf M}$, $\Phi_{\textsf S}$ & Point process of SBSs in the (MCP, SOCP)  model\\
    $\lambda_{\textsf M}$, $\lambda_{\textsf S}$ & Density of the  (MCP $\Phi_{\textsf M}$, SOCP $\Phi_{\textsf S}$)\\
    $\Phi_{\textsf M^o}$, $\lambda_{\textsf M^o}$ & Parent process for the MCP model, its density\\
    $c_{\textsf M}$, $R_{\textsf M}$ & (Mean number of points, average radius)  of each cluster in the MCP model\\
    $\Phi_{\textsf S^o}$, $\lambda_{\textsf S^o}$ & Parent process for the SOCP model, its density\\
    $c_{\textsf S^\prime}$, $R_{\textsf S^\prime}$ & (Mean number of points, average radius) of the first-order cluster in the SOCP model\\
    $c_{\textsf S}$, $R_{\textsf S}$ & (Mean number of points,  average radius) of the second-order cluster in the SOCP model\\

    \bottomrule

  \end{tabular}

\end{table}

\subsection{Network Models}
Consider a downlink HCN consisting of MBSs and SBSs randomly distributed in the horizontal plane. The processes of MBSs and SBSs are denoted as $\Phi_{m}$  with density $\lambda_m$ and $\Phi_{s}$ with density $\lambda_s$, respectively. In order to characterize the intra-tier  and inter-tier BS interdependence, the SBSs are modeled as  a cluster process distributed either as the  MCP $\Phi_{\textsf M}$ with density $\lambda_{\textsf M}$ or as the SOCP $\Phi_{\textsf S}$ with density $\lambda_{\textsf S}$, which are defined in Appendix~\ref{App:ClusterPP}. The corresponding network models are called the MCP and the SOCP models as illustrated in Fig.~\ref{Fig:Network}. In the MCP model, the MBSs are distributed as a PPP independent of the SBS process $\Phi_{\textsf M}$, which accounts for only the intra-tier BS interdependence. In contract, both the intra-tier  and inter-tier BS interdependence are captured in the SOCP model where the MBSs form the parent points in the SOCP process $\Phi_{\textsf S}$ modeling SBSs. In addition, a baseline network model, called the \emph{PPP model},  is constructed by using the PPP $\Phi_{\textsf P}$ to model the SBSs instead of $\Phi_{\textsf M}$ or $\Phi_{\textsf S}$.

A typical user is called a typical \emph{macro-cell user} (MU), denoted as $U_m$,  or a typical  \emph{small-cell user} (SU), denoted as $U_s$, depending on whether the serving BS is a MBS, $X_m$, or a SBS, $X_s$. Due to the intra- and inter-tier interdependence, it is difficult  to calculate the exact  serving distance distribution between the typical user and its serving BSs \cite{HCN-Dependence}. For tractability, we follow \cite{HCN-Dependence} in defining  the \emph{association region} for a particular MBS (or SBS) as the region  in which all the users are associated with  the MBS (or SBS) and approximating it as  a  circular region centered at the serving MBS (or SBS) with radius $D_m$ (or  $D_s$).  The MUs and SUs are uniformly distributed in the corresponding association regions  and thus the probability density function (PDF) of the serving distance is given as 
\begin{equation}
f(r)=\begin{cases}
\frac{2r}{D^{2}}, & r\leq D,\\
0, & \text{otherwise},
\end{cases}\label{eq:pdf_r}
\end{equation}
where $D=D_{m}$ for the typical MU and $D=D_{s}$ for the typical SU.  Though it is based on approximation, the above model  does provide a sufficiently  accurate description of the  stochastic distribution  of the distance between a user and its serving BS. The expressions of $D_m$ and $D_s$ for MCP and SOCP models are given in Appendix \ref{RadiusOfAssociationArea}.

\begin{figure}[t!]
\centering
\subfigure[MCP Model ]{\includegraphics [scale=0.5]{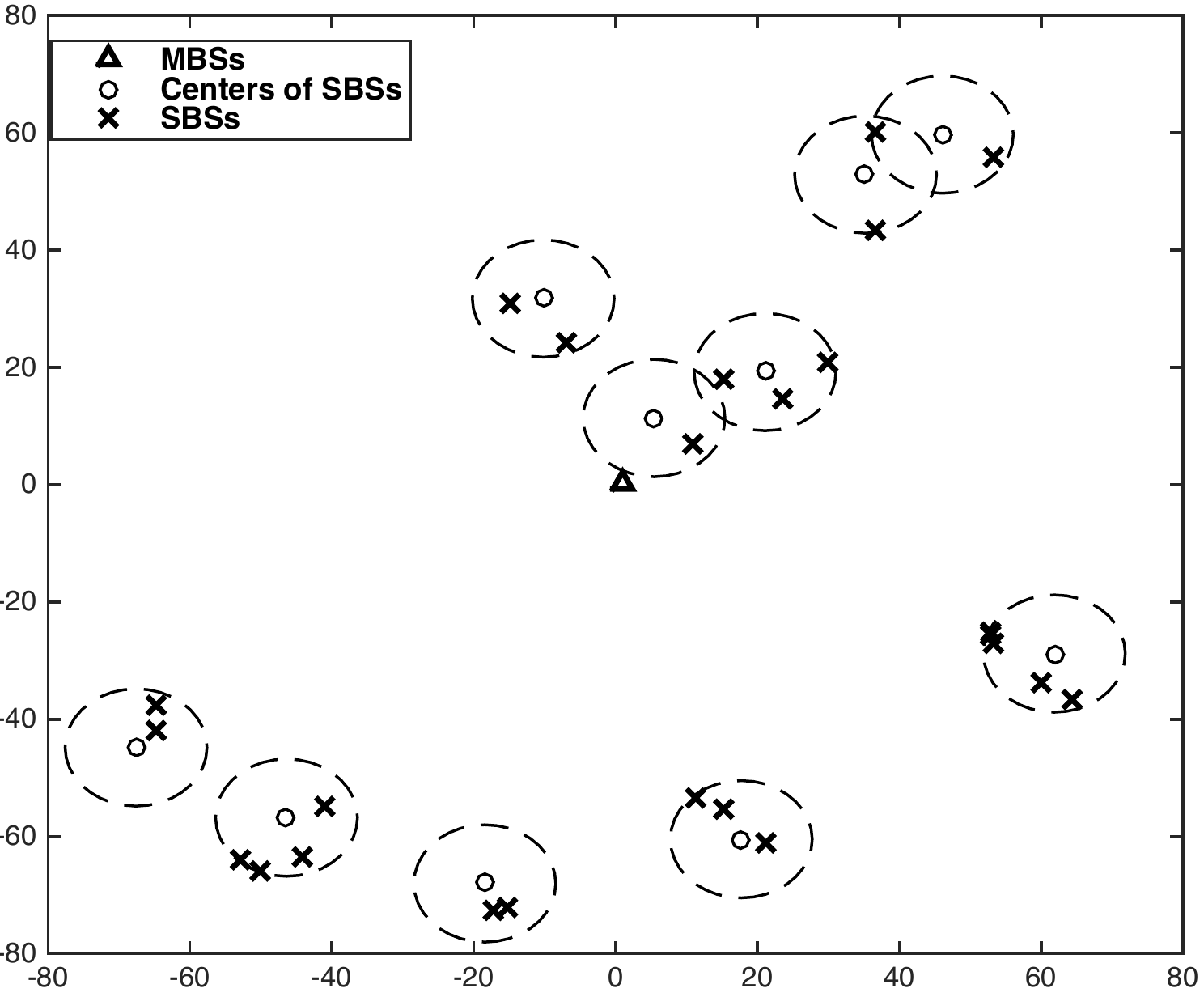}}
\subfigure[SOCP Model ]{\includegraphics[scale=0.5]{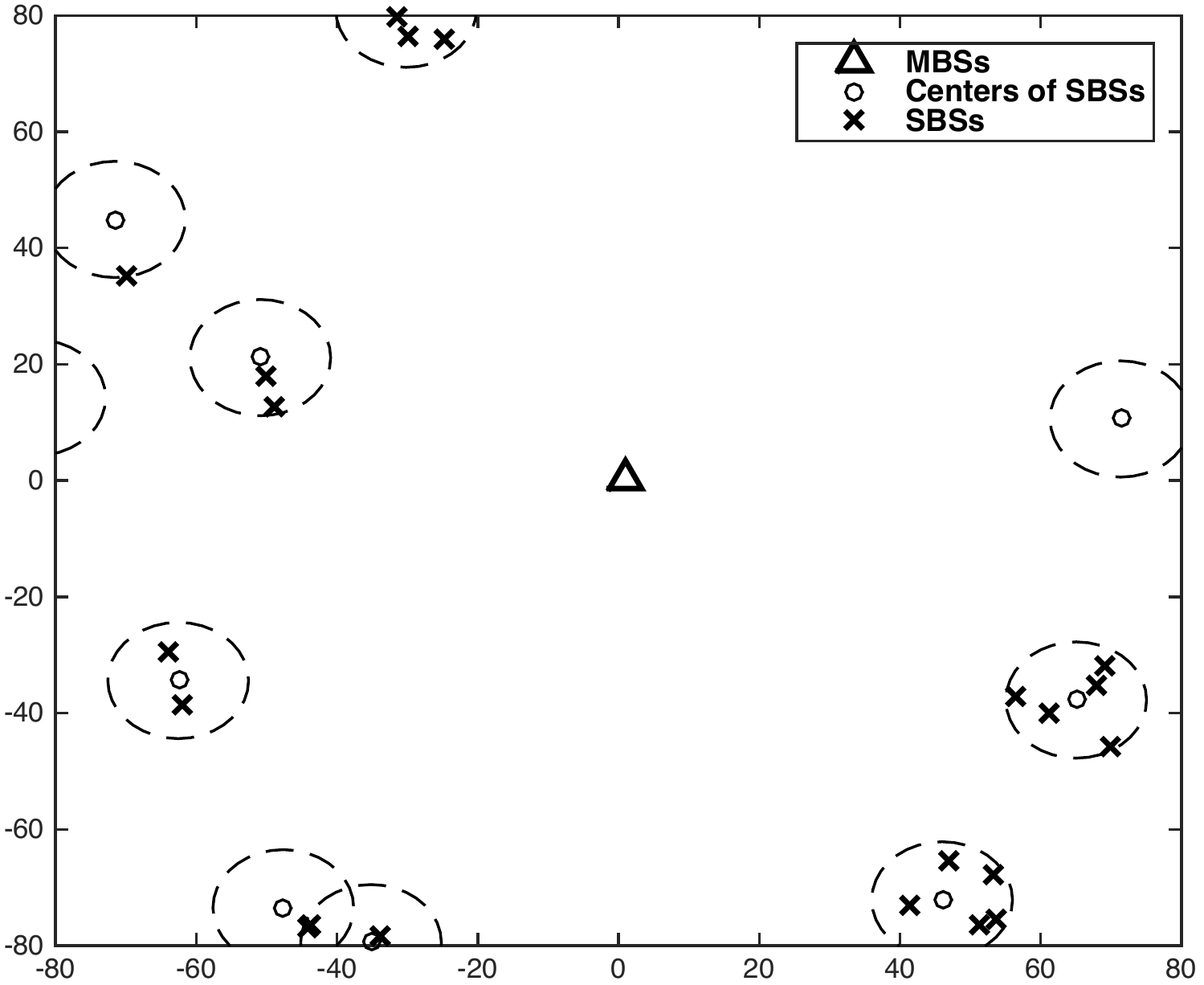}}
\caption{The HCN network model. (a) The MCP model with $\lambda_{m}=0.0001$, $\lambda_{\textsf M^o}=0.001$, $c_{\textsf M}=3$, $R_{\textsf M}=10$. (b) The SOCP model with $\lambda_{m} = \lambda_{\textsf S^o}=0.0001,$ $c_{\textsf S^\prime}=10$, $R_{\textsf S^\prime}=90$, $c_{\textsf S}=3$, $R_{\textsf S}=10$. }
\label{Fig:Network}
\end{figure}

The commonly used backlogged assumption is made in this paper, i.e., all BSs in the network are active. Note that in practice BS transmissions may be bursty  \cite{zhong2016delay,zhong16Stability} and studying the interference correlation given bursty traffic is an interesting direction for future investigation but  outside the scope of this paper. The channel model is described as follows.  MBSs and SBSs transmit at fixed power $P_m$ and $P_s$, respectively. The power received at a user at $U\in \mathds{R}^2$  in time slot $t$ due to transmission by a  BS located at $X\in \mathds{R}^2$, is given by $P h_{X, U}(t) g(X-U)^{-\alpha}$, where $g(X)$, for tractability, is the commonly used singular path-loss function $g(X) = |X|^{-\alpha}$, $|X|$  denotes the Euclidean distance from $X$ to the origin, $\alpha$ is the path-loss exponent, $P$ is $P_m$ (or $P_s$) for MBSs (or SBSs), and $h_{X, U}(t)$ denotes the Rayleigh fading process  with unit mean. For tractability, the channel fading is assumed to be temporally and spatially independent, corresponding to the environment with rich scattering and sufficiently high mobility.  In practice, channel correlation in time and space may exist but modeling it  makes the analysis intractable, which is thus omitted for simplicity. As a result,  the interference correlation in the current model arise mostly from  the correlation of BS locations. Based on the channel model, the expressions for interference power at a typical user can be obtained as follows.  We assume that all BSs transmit in the same frequency band. Consequently, there exists four types of interference: 1) from MBSs to a typical MU with power $I_{mm}=\sum_{X \in \Phi_m \setminus X_m} P_m h_{X, U_m}(t) g(X-U_m)$, 2) from MBSs to a typical SU with power $I_{ms}=\sum_{X \in \Phi_m} P_m h_{X, U_s}(t) g(X-U_s)$, 3) from SBSs to a typical MU with power $I_{sm}=\sum_{X \in \Phi_s} P_s h_{X, U_m}(t) g(X-U_m)$, 4) from SBSs to a typical SU with power $I_{ss}=\sum_{X \in \Phi_s \setminus X_s} P_s h_{X, U_s}(t) g(X-U_s)$. HCNs are usually  interference limited and thus noise is assumed to be negligible.

Time is slotted and transmission of a data packet span a single slot. The transmission of a packet is said to be successful if the received SIR at the typical  user exceeds a fixed  SIR threshold, denoted as  $\beta_m$ for MUs and $\beta_s$ for SUs. Type-I HARQ is adopted to enhance the transmission reliability. Specifically, if a transmission fails, the BS will retransmit the same packet to its user until the transmission succeeds or the maximum number of transmissions $N_{\max}$ is reached.

\subsection{Metrics}

The first part of the paper focuses on  the effects of BS interdependence (clustering) on interference correlation. It is difficult to derive the correlation coefficient for the aggregate interference from both SBSs (a MCP or SOCP) and MBSs (a PPP), which can also obscure the insight into the clustering effects of the former. Thus, for tractability and to gain simple insight, our analysis focuses on deriving the correlation coefficient for the general scenario of a single interferer field  distributed as  either the  MCP $\Phi_M$ or   the SOCP $\Phi_S$. Based on its definition, the coefficient is independent of the transmission power of   interferers that is thus assumed to be unit without loss of generality. However, in the second part of the paper focusing on network-performance analysis, different transmission powers for SBSs and MBSs are considered. Let $I(U, t)$ denote the interference power measured at the location $U \in \mathds{R}^2$ in slot $t$. Then $I(U, t) =\sum_{X \in \Phi}  h_{X}(t) g(X-U)$ where unit-transmission power is assumed without loss of generality. Due to singularity of the function $g(X)$ at the origin, the first and second moments $I(U, t)$ do not exist, which, however, are needed in quantifying the interference correlation. To overcome this difficulty, we follow the technique in \cite{InterferenceCorreLetter, InterfCorreThreeSources} by defining  $g_\epsilon(X) = \frac{1} {\epsilon+|X| ^{-\alpha}}$ with $\epsilon > 0$ such that  $g(X) = \lim_{\epsilon\rightarrow 0} g_\epsilon(X)$. Let $I_{\epsilon}(U, t)$ denote  $I(U, t)$ but with $g(X)$ replaced  by  $g_{\epsilon}(X)$. Based on the above notations, the \emph{interference correlation coefficient} that quantifies the interference correlation is denoted as $\zeta$ and defined
as the  normalized covariance of interference power  \cite{InterferenceCorreLetter, InterfCorreThreeSources}:
\begin{equation}\label{eq:CorrCoef_Definition}
\zeta(U_1, U_2, t_1, t_2)=\lim_{\epsilon\rightarrow 0} \frac{\mathbb{E}[I_{\epsilon}(U_1,t_1), I_{\epsilon}(U_2,t_2)]-\mathbb{E}[I_{\epsilon}(U_1,t_1)]\mathbb{E}[I_{\epsilon}(U_2,t_2)]}{\sqrt{\var(I_{\epsilon}(U_1,t_1))}\cdot\sqrt{\var(I_{\epsilon}(U_2,t_2))}},
\end{equation}
where $(U_1, t_1) \neq (U_2, t_2)$. Note that the interference power in the current scenario is only an approximation of that in the HCN due the omission of interference from MBSs and the Palm distribution (namely the conditioning on the given location of the typical user) for tractability.

The second  part of the paper focuses on network performance analysis. Given the Type-I HARQ transmission scheme, a suitable performance metric (see e.g., \cite{IntefCorrDiversityPolynomials}), called \emph{joint success probability} is adopted. It is denoted as $\mathcal{P}^{(n)}$ and defined as the probability that the typical user successfully receives the packets from its serving BS at $X$ within  $n$ successive transmissions. Mathematically,
\begin{equation}
\mathcal{P}^{(n)} = \mathbb{P}\l({\rm SIR}(X, t_1)>\beta,{\rm SIR}(X, t_2)>\beta, \cdots, {\rm SIR}(X, t_n)>\beta\r),
\end{equation}
where ${\rm SIR}(X,t_n)$ denotes the SIR received at the typical user in slot $n$ and $\beta$ is a fixed threshold. Note that the metric can be translated into the  \emph{delay-limited throughput} (see e.g., \cite{HARQ-CorrInterf-adhoc}) or \emph{transmission capacity} measuring  network spatial throughput (see e.g., \cite{StoGeo}).

\section{Analysis of the Interference-Correlation Coefficient}
In this section, we analyze  the interference-correlation coefficient for the scenario where the interferers are distributed as a cluster point process, namely either MCP or SOCP. It is shown that interferer clustering enhances the interference correlation.

To this end, the first and second moments of interference power are derived as shown in the following two lemmas.

\begin{lemma}\label{lm:MeanInterference}
\emph{
The expectations  of the interference power  $I_{\epsilon}(U,t)$, called \emph{mean interference},  for both the MCP  and SOCP models have an identical expression given as
\begin{equation}\label{eq:MeanInterf}
\mathbb{E}[I_{\epsilon}(U,t)]= \mathbb{E}[h] \lambda \int_{\mathbb{R}^{2}}g_{\epsilon}(X)\mathrm{d}X,
\end{equation}
where $\lambda = \lambda_{\textsf M}$ for  the MCP model and $\lambda = \lambda_{\textsf S}$ for the SOCP model.}
\end{lemma}

\begin{proof}
See Appendix \ref{pf:MeanInterference}.
\end{proof}

\begin{remark}[Comparison with the PPP Model]\emph{It is interesting to note that the expression in \eqref{eq:MeanInterf}  also holds for the mean interference for the PPP model where $\lambda$ is then the density of the PPP \cite{InterferenceCorreLetter}. In other words, the mean interference is invariant to point clustering.}
 \end{remark}
For ease of notation, define two functions as below, which are used for stating the results in Lemma~\ref{lm:SecondMoment} in the sequel:
\begin{align}
F(c,R)&=\frac{c}{\pi^{2}R^{4}}\int_{\mathbb{R}^{2}}\int_{\mathbb{R}^{2}}g(X)g(Y)A_{R}(|X-Y|)\mathrm{d}X\mathrm{d}Y, \label{Eq:Fun1}\\
A_{R}(r)&=\left\{\begin{aligned}
&2R^{2}\arccos\left(\frac{r}{2R}\right)-r\sqrt{R^{2}-\frac{r^{2}}{4}}, && 0\leq r\leq2R, \\
&0, && \text{otherwise},
\end{aligned}\right. \label{Eq:Fun2}
\end{align}
where $(\lambda, c, R)$ is equal to $(\lambda_{\textsf M},  c_{\textsf M},  R_{\textsf M})$ and $(\lambda_{\textsf S}, c_{\textsf S}, R_{\textsf S})$ for  the MCP and SOCP models, respectively. Although the function $F(c,R)$ can not be written in the closed-form expression, it can be numerically calculated with standard numeric software, such as Matlab. Further, the following lemma provides the approximation of $F(c,R)$ for large $R$ in closed-form.

\begin{lemma}\label{lm:Approx_F}
\emph{For large $R$, $F(c,R)$ is given as
\begin{equation}\label{Eq:Approx_F}
F(c,R) = \frac{4 c \pi^3 }{\alpha^2 R^2} \epsilon^{{4-2 \alpha}/\alpha} (\csc({2 \pi}/\alpha))^2 + o\;(1/R^2)
\end{equation}
}
\end{lemma}

\begin{proof}
See Appendix \ref{pf:Approx_F}
\end{proof}

Lemma \ref{lm:Approx_F} provides a simpler method to approximately calculate the interference correlation coefficient in Theorem~\ref{th:IntCorrCoe} for large $R$. Further, the numerical results show that the approximation of $F(c,R)$ evaluates the interference correlation coefficient well, even in the case of small $R$.

\begin{lemma}\label{lm:SecondMoment}
\emph{
The mean product between the interference power $I_{\epsilon}(U_1,t_1)$ and $I_{\epsilon}(U_2,t_2)$ for both  the MCP and SOCP models is given by
\begin{equation} \label{eq:MeanProd}
\mathbb{E}[I_{\epsilon}(U_1,\!t_1), I_{\epsilon}(U_2,\!t_2)] \!\!=\!\! \mathbb{E}[h]^2 \lambda \!\!\left[ \int_{\mathbb{R}^2} \!\!g_{\epsilon}(X\!\!-\!\!U_1) g_{\epsilon}(X\!\!-\!\!U_2) \mathrm{d} X \!\!+\!\! \lambda \!\left( \!\int_{\mathbb{R}^2} \!\!g_{\epsilon}(X) \mathrm{d} X \!\!\right)^2+\!\! F(c,\!R)\right]\!\!,
\end{equation}
and the second moment of interference power  is given as
\begin{equation}\label{eq:SecondMoment}
 \mathbb{E}[I_{\epsilon}^2(U,t)]
 =\mathbb{E}[h^2]\lambda \! \int_{\mathbb{R}^{2}}\!g_{\epsilon}^2(X)\mathrm{d}X \!+ \! \mathbb{E}[h]^{2}\lambda \!\left[\lambda \left(\!\int_{\mathbb{R}^{2}}\!g_{\epsilon}(X)\mathrm{d}X
  \!\right)^2+ \!F(c,R)\right],
\end{equation}
 where the function $F(\cdot, \cdot)$ is given in \eqref{Eq:Fun1}.
}
\end{lemma}

\begin{proof}
See Appendix \ref{pf:SecondMoment}.
\end{proof}

\begin{remark}[Comparison with the PPP Model]
\emph{Consider the PPP model, the interference mean product and second moment are given as \cite{InterferenceCorreLetter}:
\begin{align}
\mathbb{E}[I_{\epsilon}(U_1,\!t_1),\!I_{\epsilon}(U_2,\!t_2)] &= \mathbb{E}[h]^2 \lambda \!\!\left[ \int_{\mathbb{R}^2} \!\!g_{\epsilon}(X\!\!-\!\!U_1) g_{\epsilon}(X\!\!-\!\!U_2) \mathrm{d} X \!\!+\!\! \lambda \left ( \!\int_{\mathbb{R}^2} \!\!g_{\epsilon}(X) \mathrm{d} X  \right)^2\right], \label{eq:MeanProd_PPP}\\
 \mathbb{E}[I_{\epsilon}^2(U,t)] &=\mathbb{E}[h^2]\lambda \! \int_{\mathbb{R}^{2}}\!g_{\epsilon}^2(X)\mathrm{d}X \!+ \! \mathbb{E}[h]^{2}\lambda \!\left[\lambda \left(\!\int_{\mathbb{R}^{2}}\!g_{\epsilon}(X)\mathrm{d}X
  \!\right)^2\right]. \label{eq:SecondMoment_PPP}
\end{align}
Comparing these results with those in Lemma~\ref{lm:SecondMoment}, both the interference mean product and second moment for  the MCP and SOCP models are greater than their counterparts for  the PPP model. This shows that the interferer  clustering  changes the interference distribution and thereby enhances the interference correlation as shown shortly.}
\end{remark}

Using  Lemma \ref{lm:MeanInterference} and Lemma \ref{lm:SecondMoment}, the interference correlation coefficient is derived by substituting \eqref{eq:MeanInterf}, \eqref{eq:MeanProd} and \eqref{eq:SecondMoment} into \eqref{eq:CorrCoef_Definition}, yielding the following theorem.
\begin{theorem} \label{th:IntCorrCoe}
\emph{
The spatial-and-temporal interference correlation coefficient for the MCP model, namely  $\zeta_{\textsf M}$,  and that for the SOCP model, namely  $\zeta_{\textsf S}$, can be both written as:
\begin{equation}\label{eq:CorrCoef_MCP}
\zeta(U_1, U_2, t_1, t_2) =\lim_{\epsilon\rightarrow 0}\frac{\int_{\mathbb{R}^{2}}g_{\epsilon}(X)g_{\epsilon}(X- U_1+U_2)\mathrm{d}X+F(c,R)}{\frac{\mathbb{E}\left[h^{2}\right]}{\mathbb{E}\left[h\right]^{2}}\int_{\mathbb{R}^{2}}g_{\epsilon}^{2}(X)\mathrm{d}X+F(c,R)},
\end{equation}
where $(c, R)$ is equal to $(c_{\textsf M}, R_{\textsf M})$ and $(c_{\textsf S}, R_{\textsf S})$  for the MCP and SOCP models, respectively.}
\end{theorem}
Theorem~\ref{th:IntCorrCoe} shows that the interference-correlation coefficients for the MCP and SOCP models are identical if their parameters match, namely $(c_{\textsf M}, R_{\textsf M}) = (c_{\textsf S}, R_{\textsf S})$. Note that these coefficients depend only on the first and second moments of the interference distributions but the network-performance metric, namely the joint success probability, depends on the higher moments. For this reason, despite the mentioned equivalence in interference correlation, the joint success probabilities for the two models differ as shown in the next section.

\begin{remark}[Comparison with the PPP Model]  \label{Rmk:CorCoef_PPP}
\emph{
The  interference correlation coefficient in the PPP model is given as \cite{InterferenceCorreLetter}
\begin{equation}
\zeta_{\textsf P}(U_1, U_2, t_1, t_2)=\lim_{\epsilon \rightarrow 0 }\frac{\int_{\mathbb{R}^{2}}g_{\epsilon}(X)g_{\epsilon}(X- U_1+U_2)\mathrm{d}X}{\frac{\mathbb{E}\left[h^{2}\right]}{\mathbb{E}\left[h\right]^{2}}\int_{\mathbb{R}^{2}}g_{\epsilon}^{2}(X)\mathrm{d}X}.\label{eq:CorrCoef_PPP}
\end{equation}
Comparing $(\zeta_{\textsf M}, \zeta_{\textsf S})$ in Theorem~\ref{th:IntCorrCoe} with $\zeta_{\textsf P}$, one can observe that the effect of interference clustering on the  interference-correlation coefficient  is characterized by the function $F(c,R)$ given in \eqref{Eq:Fun1}, which depends  on the mean number of points and the radius of each cluster. The mathematical comparison between the interference-correlation coefficients is provided in the following proposition.
}
\end{remark}

\begin{proposition}\label{Prop:CorCoef_MCP_PPP}
\emph{The interference-correlation coefficients for the MCP and SOCP  models are greater than that for the PPP model: $\zeta_{\textsf M}\geq \zeta_{\textsf P}$ and $\zeta_{\textsf S}\geq\zeta_{\textsf P}$,  where the equalities hold when the cluster parameters satisfy $F(c,R)=0$. }
\end{proposition}
\begin{IEEEproof}
See Appendix \ref{pf:CC_MCP_PPP}.
\end{IEEEproof}

Proposition~\ref{Prop:CorCoef_MCP_PPP} shows that BS clustering increases the level of interference temporal  correlation. Based on the relation derived in prior work \cite{LocalDelay_PoissonNet}, this can result in growing \emph{local delay}, defined the expected number of time slots required for the successful transmission of a packet. Next, the relations between the level of interference correlation and the cluster parameters are specified in the following proposition.

\begin{proposition}\label{Prop:CorCoef_c_R}
\emph{The interference correlation coefficients $\zeta_{\textsf M}$ and $\zeta_{\textsf S}$ are monotone increasing functions of  $c_{\textsf M}$ and $c_{\textsf S}$, respectively, and monotone decreasing functions of  $R_{\textsf M}$ and  $R_{\textsf S}$, respectively. Furthermore, $\zeta_{\textsf M} \rightarrow \zeta_{\textsf P}$ as $\frac{c_{\textsf M}}{R_{\textsf M}^{2}}\rightarrow  0$ and $\zeta_{\textsf S} \rightarrow \zeta_{\textsf P} $ as $\frac{c_{\textsf S}}{R_{\textsf S}^{2}}\rightarrow 0$.}
\end{proposition}

\begin{proof}
See Appendix \ref{pf:IntCorCoef_c_R}.
\end{proof}
Consider the MCP model without loss of generality. Both reducing $R_{\textsf M}$ for a fixed $c_{\textsf M}$ and increasing $c_{\textsf M}$ for a fixed $R_{\textsf M}$ increase the interferer density in each cluster and thus the level of clustering, leading to the results in Proposition~\ref{Prop:CorCoef_c_R}. As  $\frac{c_{\textsf M}}{R_{\textsf M}^{2}}\rightarrow  0$ and $\frac{c_{\textsf S}}{R_{\textsf S}^{2}}\rightarrow 0$,  Proposition \ref{Prop:CorCoef_c_R} suggests that the effects of interferer  clustering on interference correlation can be neglected since the interference-correlation coefficients for the cluster interferer processes converge to that of the PPP without clustering.

\section{Performance of HCNs with Clustered Small Cells}

In this section, to investigate  the effects of spatial BS interdependence on network performance, we analyze  the joint success probabilities  for HCNs with clustered SBSs.

\subsection{Joint Success Probability}
First, the \emph{conditional} joint success probability is derived in Lemma~\ref{lm:ConJointSuccProb}, which is conditioned on the fixed distance between the typical user and the  serving BS. To this end, some useful functions are defined as follows. Let  $G_{\Phi}[v(X)]\triangleq \mathbb{E}\left(\prod_{X\in \Phi}v(X)\right)$ denotes the \emph{probability generating functional} (PGF) of a general  point process $\Phi$ where the operator $\mathbb{E}$ is the expectation with respect to the Palm measure of  $\Phi$ and $v$ with $ 0 \leq v \leq 1$ is a bounded measurable function. Let $\mathbb{E}_{X_o}^{!}$ denote the expectation operator with respect to the reduced Palm measure of $\Phi$,  which is the conditional expectation over $\Phi\backslash \{X_0\}$ given  a point  $X_o\in \Phi$ being  fixed \cite{PoissonClusterProcess_TIT}. Using this definition, the conditional PGF of the point process $\Phi$ is defined as
$G_{\Phi^{!}}\left(X\right)\triangleq \mathbb{E}_{X_o}^{!}\left(\prod_{X\in \Phi}v(X)\right)$. Based on the above notations and definitions, the conditional joint success  probabilities are obtained as shown in the following lemma.

\begin{lemma}\label{lm:ConJointSuccProb}
\emph{Consider a HCN allowing retransmissions over $n$ slots. Given the propagation  distance $r_m$ for the typical MU
and $r_s$ for the typical SU, the conditional joint success probabilities for the MU and SU, denoted as $\mathcal{P}_{m}^{\left(n\right)}(r_m)$ and $\mathcal{P}_{s}^{\left(n\right)}(r_s)$, respectively, are given as:
\begin{align}
\mathcal{P}_{m}^{\left(n\right)}(r_m) &=
G_{\Phi_{m}^{!}}\left[\left(1+\frac{\beta_{m}\widetilde{g}(X,r_m)} {r_m^{-\alpha}}\right)^{-n}\right]G_{\Phi_{s}}\left[\left(1+\frac{\beta_{m}P_{s}|X|^{-\alpha}}{P_{m} r_m^{-\alpha}}\right)^{-n}\right],\label{eq:JSP_MU_original}\\
\mathcal{P}_{s}^{\left(n\right)}(r_s) &=G_{\Phi_{m}}\left[\left(1+\frac{\beta_{s}P_{m} |X|^{-\alpha}}{P_{s} r_s^{-\alpha}}\right)^{-n}\right]G_{\Phi_{s}^{!}}\left[\left(1+\frac{\beta_{s} \widetilde{g}(X,r_s)}{r_s^{-\alpha}}\right)^{-n}\right], \label{eq:JSP_FU_original}
\end{align}
where the MBS process $\Phi_{m}$ is the PPP $\Phi_{\textsf P}$ and the SBS process $\Phi_s = \Phi_{\textsf M}$ in the MCP model and $\Phi_s = \Phi_{\textsf S}$ in the SOCP model, $\widetilde{g}(X,r)=|X|^{-\alpha}\mathds{1}(|X|>r)$.
}
\end{lemma}

\begin{IEEEproof}
See Appendix \ref{pf:Cond_JSP}.
\end{IEEEproof}
The expressions for the  PGFs and conditional PGFs in Lemma~\ref{lm:ConJointSuccProb} for specific point processes can be found in e.g., \cite{StoGeoBook-Martin, PoissonClusterProcess_TIT, HCN_SOCP}, and are provided in the following lemma.

\begin{lemma}\label{lm:PGF}\emph{
The PGFs and the conditional PGFs for a PPP, MCP and SOCP are given as follows:
\begin{itemize}
\item (PPP)  \cite{StoGeoBook-Martin}
\begin{equation}
G_{\Phi_{\textsf P}}(v)=G_{\Phi_{\textsf P}^{!}}(v)=\exp\left[-\lambda\int_{\mathbb{R}^{2}}\left(1-v(X)\right)\mathrm{d}X\right]; \label{eq:PGF_PPP}
\end{equation}
\item (MCP)  \cite{PoissonClusterProcess_TIT}
\begin{align}
G_{\Phi_{\textsf M}}(v) &=\exp\left(-\lambda_{\textsf M^o}\int_{\mathbb{R}^{2}}\left[1-M\left(\int_{\mathbb{R}^{2}}v(X+Y)f_{\textsf M}(Y)\mathrm{d}Y\right)\right]\mathrm{d}X\right),\label{eq:PGF_MCP}\\
G_{\Phi_{\textsf M}^{!}}(v) &=G_{\Phi_{MCP}}(v)\int_{\mathbb{R}^{2}}M\left(\int_{\mathbb{R}^{2}}v(X+Y)f_{\textsf M}(X)\mathrm{d}X\right)f_{\textsf M}(Y)\mathrm{d}Y,\label{eq:Conditional PGF_MCP}
\end{align}
where $M(x)=\exp(-c_{\textsf M}(1-x))$;
\item (SOCP)  \cite{HCN_SOCP}
 \begin{align}
G_{\Phi_{\textsf S}}(v) &=\exp \!\!\left[-\lambda_{\textsf S^o}\!\!\int_{\mathbb{R}^{2}}\!\!\left\{ 1-M_{1}\!\!\left[\int_{\mathbb{R}^{2}}M_{2}(\int_{\mathbb{R}^{2}}v(X\!+\!Y+\!Z)f_{\textsf S}(Z)\mathrm{d}Z)f_{\textsf S^\prime}(Y)\mathrm{d}Y\!\right]\right\} \mathrm{d}X\!\right],\label{eq:PGF_SOCP}\\
G_{\Phi_{\textsf S}^{!}}(v) & =G_{\Phi_{\textsf S}}(v)M_{1}\left[\int_{\mathbb{R}^{2}}M_{2}(\int_{\mathbb{R}^{2}}v(X+Y+Z)f_{\textsf S}(Z)\mathrm{d}Z)f_{\textsf S^\prime}(Y)\mathrm{d}Y\right]\nonumber \\
 & \cdot\int_{\mathbb{R}^{2}}M_{2}(\int_{\mathbb{R}^{2}}v(X+Y+Z)f_{\textsf S}(Z)\mathrm{d}Z)f_{\textsf S}(Y)\mathrm{d}Y. \label{eq:Coditional PGF_SOCP}
\end{align}
where $M_{1}(x)=\exp\left(-c_{\textsf S^\prime}(1-x)\right)$ and $M_{2}(x)=\exp\left(-c_{\textsf S}(1-x)\right)$,
\end{itemize}}
\end{lemma}

Last, the joint success probabilities are obtained as the expectations of the conditional probabilities in Lemma~\ref{lm:ConJointSuccProb} with respect to the distribution of the  propagation distance of the typical MU/SU. Since the user is uniformly distributed in the coverage area assumed as a disk with radius $D$, the PDF of the distance is given in (\ref{eq:pdf_r}), where $D=D_{m}$ if the typical user is a MU or  otherwise $D=D_{s}$.  Combining \eqref{eq:pdf_r} and Lemma~\ref{lm:ConJointSuccProb} leads to the following main result.

\begin{theorem}\label{Thm:JSP}
\emph{
For a HCN allowing retransmissions over $n$ slots, the joint success probabilities  for the MU and SU, denoted as $\mathcal{P}_{m}^{\left(n\right)}$
and $\mathcal{P}_{s}^{\left(n\right)}$, respectively, are given as:
\begin{align}
\mathcal{P}_{m}^{(n)} = \frac{2}{D_{m}^{2}}\int_{0}^{D_{m}} \mathcal{P}_{m}^{(n)}(r)r\mathrm{d}r, \\
\mathcal{P}_{s}^{(n)} = \frac{2}{D_{s}^{2}}\int_{0}^{D_{s}} \mathcal{P}_{s}^{(n)}(r)r\mathrm{d}r,
\end{align}
where $\mathcal{P}_{m}^{(n)}(r)$ and $\mathcal{P}_{s}^{(n)}(r)$ are provided in Lemma \ref{lm:ConJointSuccProb}.
}
\end{theorem}
Then the specific expressions for the joint success probability corresponding to the MCP and SOCP models can be derived by substituting the results in Lemma~\ref{lm:PGF} into those in Theorem~\ref{Thm:JSP}.  The results have complex  expressions with multiple integrals. This reflects  the theoretical challenge in characterizing the effects of SBS clustering in practice on the HCN performance. Nevertheless, the results obtained in this section can be leveraged in the next section to yield simple insight.

\begin{remark}
\emph{
The joint success probability in the MCP (or the SOCP)  model  is a monotone-decreasing function of $c_{\textsf M}$ (or $c_{\textsf S}$). The reason is that  increasing the mean number of points per cluster, $c_{\textsf M}$ (or $c_{\textsf S}$),  increases the interference power from the SBSs to MUs but does not change the signal strength. Note that the propagation distance of a data link depends on  the coverage radiuses  of MBSs which are independent of $c_{\textsf M}$ (or $c_{\textsf S}$) (see the system model).
}
\end{remark}

Joint success probability provides a  basic component   for further calculating  different practical network performance metrics, such as delay-limited throughput (see \cite[Eq. (4)]{HARQ-CorrInterf-adhoc}) and local delay (see \cite[Eq. (26)]{IntefCorrDiversityPolynomials})). Specifically, the metrics are linear functions of joint success probability and the calculation procedure is straightforward and omitted for brevity.

\subsection{Bounds on  Joint Success Probabilities}
 Although the expressions for the joint success probabilities are derived  in the preceding sub-section, the results have complex expressions.  In this sub-section, the  probabilities are bounded by their PPP counterparts. The results yield useful insight into the effects on SBS clustering on the network performance.

The method of bounding the joint success probabilities relies on bounds on the PGFs for the MCP and SOCP.  Throughout this section, the PGFs are considered as functions of the density, $\lambda$,  of the corresponding point process $\Phi$ while the original argument $\nu$  is identical for different point processes (see the PGF definitions in the preceding sub-section). Then the PGF and conditional PGF for the MCP can be bounded by their PPP counterparts  as shown  in \cite{PoissonClusterProcess_TIT}:
\begin{align}
G_{\Phi_{\textsf P}}\left(\lambda_{\textsf M}\right)&\leq G_{\Phi_{\textsf M}}\left(\lambda_{\textsf M}\right)\leq G_{\Phi_{\textsf P}}\left(\frac{\lambda_{\textsf M}}{1+c_{\textsf M}}\right),\label{eq:Bounds_MCP}\\
G_{\Phi_{\textsf P}^{!}}\left(\lambda_{\textsf M}+\frac{c_{\textsf M}}{\pi R_{\textsf M}^{2}}\right)&\leq G_{\Phi_{\textsf M}^{!}}\left(\lambda_{\textsf M}\right)\leq G_{\Phi_{\textsf P}^{!} }\left(\frac{\lambda_{\textsf M}}{1+c_{\textsf M}}\right),\label{eq:Bounds_Conditional_MCP}
\end{align}
where  the PGF and conditional PGF of the PPP $\Phi_{\textsf P}$ with density $\lambda$ are identical and given as
\begin{equation}
G_{\Phi_{\textsf P}}(\lambda)=G_{\Phi_{\textsf P}}(\lambda)= \exp\left(-\lambda\int_{\mathbb{R}^{2}}\left(1-v(x)\right)\mathrm{d}x\right).
\end{equation}
These results for the MCP are extended to the SOCP as shown in the following lemma.

\begin{lemma}\label{lm:bounds_PGF_SOCP}\emph{The PGF and the conditional PGF for the SOCP $\Phi_{\textsf S}$ can be bounded as
\begin{align}
G_{\Phi_{\textsf P}}\left(\lambda_{\textsf S}\right) &\leq G_{\Phi_{\textsf S}}(\lambda_{\textsf S})\leq G_{\Phi_{\textsf P}}\left(\frac{\lambda_{\textsf S}}{\left(1+c_{\textsf S^\prime}\right)\left(1+c_{\textsf S}\right)}\right),\label{eq:Bounds_SOCP}\\
G_{\Phi_{\textsf P}^{!}}\left(\lambda_{\textsf S}+c_{\textsf S^\prime}c_{\textsf S}\gamma+\frac{c_{\textsf S}}{\pi R_{\textsf S}^{2}}\right)&\leq G_{\Phi_{\textsf S}^{!}}(\lambda_{\textsf S})\leq G_{\Phi_{\textsf P}^{!}}\left(\frac{\lambda_{\textsf S}}{\left(1+c_{\textsf S^\prime}\right)\left(1+c_{\textsf S}\right)}\right) \label{eq:Bounds_Conditional_SOCP}
\end{align}
where the constant $\gamma$ is defined as
\begin{equation}\nonumber
\gamma=\min\left\{ \frac{1-\exp\left(\frac{-R_{\text S^\prime}^{2}}{2\sigma^{2}}\right)}{\pi R_{\text S^\prime}^{2}+2\pi\sigma^{2}\left(\exp\left(\frac{-R_{\text S^\prime}^{2}}{2\sigma^{2}}\right)-1\right)},\frac{1}{\pi R_{\text S}^{2}}\right\},
\end{equation}
}
\end{lemma}
\begin{proof}
See Appendix \ref{pf:Bouds_socp}.
\end{proof}

Next, consider the baseline PPP model. Let $\mathcal{P}_{m \textsf P}^{(n)}(\lambda_s)$ and $\mathcal{P}_{s \textsf P}^{(n)}(\lambda_s)$ denote the joint success probabilities for the typical MU and SU, respectively, which are functions of the SBS density $\lambda_s$. Using Lemma~\ref{lm:ConJointSuccProb}, the probabilities conditioned on a propagation distance $r$ for the corresponding typical users can be derived as shown in the following lemma.
\begin{lemma}\label{lm:ConJointSuccProb:PPP}
\emph{Consider a HCN allowing retransmissions over $n$ slots and having SBSs distributed as PPP. Given the propagation  distance $r$ between a typical user and the serving BS, the conditional joint success probabilities for the MU and SU are given as:
\begin{align}
\mathcal{P}_{m\textsf P}^{\left(n\right)}(\lambda_s, r) &=
\exp\left[-\lambda_{m}Q_{n}(\beta_{m}) r^{2}\right] \exp\left[- \lambda_s \left(\frac{\beta_{m}P_{s}}{P_{m}}\right)^{\delta}U_{n} r^{2}\right], \\
\mathcal{P}_{s \textsf P}^{\left(n\right)}(\lambda_s, r) &= \exp\left[- \lambda_m \left(\frac{\beta_{s}P_{m}}{P_{s}}\right)^{\delta}U_{n} r^{2}\right] \exp\left[-\lambda_s Q_{n}(\beta_{s}) r^{2}\right],
\end{align}
where the function $Q_n(\beta)$ of a SIR threshold $\beta$ is given as
\begin{equation}
Q_{n}(\beta)=\pi\delta\sum_{m=1}^{n}  \binom {n} {m}
\frac{\left(-1\right)^{m+1}\beta^{m}}{m-\delta}{}_{2}F_{1}(m,m-\delta;m-\delta+1;-\beta), \label{eq:Qn}
\end{equation}
with $ \delta=2/\alpha$ and the constant
\begin{equation}
U_{n}=\frac{\pi^{2}\delta}{\sin\left(\pi\delta\right)} \frac{\Gamma(n+\delta)}{\Gamma(n) \Gamma(1+\delta)}. \label{eq:Un}
\end{equation}
}
\end{lemma}
By taking expectation  with respect to the distance distribution in \eqref{eq:pdf_r},
the joint success probabilities for the PPP model follow from Lemma~\ref{lm:ConJointSuccProb:PPP} as shown below.

\begin{lemma}\label{Lem:JointSuccProb:PPP} \emph{Consider a HCN allowing retransmissions over $n$ slots and having SBSs distributed as a PPP. The conditional joint success probabilities for the MUs and SUs are given as:
\begin{align}
\mathcal{P}_{m \textsf P}^{(n)}(\lambda_s) & = D_{m}^{-2}\!\left[\lambda_{m}Q_{n}(\beta_{m})\!+\!\lambda_s \!\left(\!\frac{\beta_{m}P_{s}}{P_{m}}\!\right)^{\delta}\!\!U_{n}\!\right]^{-1}, \\
\mathcal{P}_{s \textsf P}^{(n)}(\lambda_s) &= D_{s}^{-2}\!\left[\lambda_{m}\!\left(\!\frac{\beta_{s}P_{m}}{P_{s}}\!\right)^{\delta}\!\!U_{n}+ \lambda_s Q_{n}(\beta_{s})\!\right]^{-1}.
\end{align}
}
\end{lemma}

\begin{proof}
See Appendix \ref{Pf:JSP_PPP}.
\end{proof}

Last, using Lemmas \ref{lm:bounds_PGF_SOCP} to \ref{Lem:JointSuccProb:PPP}, the main results of this sub-section are derived and presented in the following theorem.

\begin{theorem}\label{Thm:Bounds_JSP}  \emph{Consider a HCN allowing retransmissions over $n$ slots. The joint success probabilities for the MPC and SOCP models can be bounded by their counterparts for the PPP model as follows:
\begin{itemize}
\item For the MCP model,
\begin{align}
\mathcal{P}_{m \textsf P}^{(n)}(\lambda_{\textsf M}) &\leq \mathcal{P}_{m \textsf M}^{(n)}
\leq \mathcal{P}_{m \textsf P}^{(n)}\l(\frac{\lambda_{\textsf M}}{1+c_{\textsf M}}\r). \\
\mathcal{P}_{s \textsf P}^{(n)}(\lambda_{\textsf M}+\frac{c_{\textsf M}}{\pi R_{\textsf M}^2})
&\leq \mathcal{P}_{s \textsf M}^{(n)}
\leq \mathcal{P}_{s \textsf P}^{(n)}\l(\frac{\lambda_{\textsf M}}{1+c_{\textsf M}}\r).
\end{align}
\item For the SOCP model,
\begin{align}
\mathcal{P}_{m \textsf P}^{(n)}(\lambda_{\textsf S})
&\leq \mathcal{P}_{m \textsf M}^{(n)}(\lambda_{\textsf S})
\leq \mathcal{P}_{m \textsf S}^{(n)}
\leq \mathcal{P}_{m \textsf P}^{(n)}(\frac{\lambda_{\textsf S}}{(1+c_{\textsf S^\prime})(1+c_{\textsf S})}), \label{eq:Compare_JSP_MU}\\
\mathcal{P}_{s \textsf P}^{(n)}(\lambda_{\textsf S}+ c_{\textsf S^\prime} c_{\textsf S} \gamma+ \frac{c_{\textsf S}}{\pi R_{\textsf S}^2})
&\leq \mathcal{P}_{s \textsf S}^{(n)}
\leq \mathcal{P}_{s \textsf P}^{(n)}(\frac{\lambda_{\textsf S}}{(1+c_{\textsf S^\prime})(1+c_{\textsf S})}),
\end{align}
\end{itemize}
where $\mathcal{P}_{m \textsf P}^{(n)}(\lambda)$ and $\mathcal{P}_{s \textsf P}^{(n)}(\lambda)$ for the PPP model are given in Lemma~\ref{Lem:JointSuccProb:PPP}. }
\end{theorem}

\begin{proof}
See Appendix \ref{Pf:Bounds_JSP}.
\end{proof}

Theorem~\ref{Thm:Bounds_JSP} shows that the joint success probability for the typical MU is increasing in the order of the PPP, MCP and SCOP models. This suggests that increasing the level of BS inter-dependence improves the MU's performance.

\begin{remark}
\emph{
Theorem \ref{Thm:Bounds_JSP} mathematically shows that the joint success probability in the MCP model and SOCP model converge to that in the PPP model when the mean number points in each cluster, i.e., $c_{\textsf M}$, $c_{\textsf S^\prime}$, $c_{\textsf S}$, approximates to 0. This is because, the upper bound and lower bound converge to the joint success probability in the PPP model under the above condition.
}
\end{remark}

It is inferred that, comparing with the case of independent interference, the interference correlation degrades the performance of HCNs with retransmission, and the degradation increases as the attraction between BSs increases. This is because, interference correlation reduces the diversity gain in retransmission (see \cite{IntefCorrDiversityPolynomials,InterfCorreDiversityLoss,CorrelationMobileRandomNet,LocalDelay_MAC}) and the interference correlation increases as the attraction between BSs increases (see Proposition \ref{Prop:CorCoef_c_R}). Therefore, correlation-aware retransmission scheme is needed to improve the network performance. Based on the observations obtained in this paper, we propose a correlation-aware retransmission scheme as follows.

\begin{remark}[Correlation-aware Retransmission Scheme] \label{Rm:CorrRetrans} \emph{For each cluster, if most nodes transmit successfully (i.e., success probability is larger than a given threshold) in the current time slot, all nodes in the cluster will transmit in the next time slot to take advantage of the interference (success) correlation. In contrast, if most nodes fail (i.e., success probability is lower than a given threshold), only the successful nodes transmit in the next time slot and the unsuccessful nodes keep silence for a randomly chosen time slots to reduce the interference correlation and avoid the blind retransmission. In particular, if all the transmissions fail in the current time slot, the nodes will be randomly chosen to transmit or not in the next time slot.}
\end{remark}

The proposed correlation-aware retransmission scheme takes good advantage of interference correlation when the success probability is high. This is because high success probability and interference correlation means there is a high probability that the transmission will succeed in the next time slot. Furthermore, the proposed scheme effectively avoid the blind retransmission when the success probability is low. The simulation results in the sequel show that the proposed scheme significantly improve the success probability in HCNs.

\section{Numerical Results}

\subsection{Interference Correlation}

In this subsection, the interference-correlation coefficients are evaluated for the PPP, MCP and SOCP models to illustrate their relation and the effects of system parameters. For fair comparison, the parameters are set as follows: $\lambda_{\textsf P}=\lambda_{\textsf M}=\lambda_{\textsf S}$, $c_{\textsf M}=c_{\textsf S}$, and $R_{\textsf M}=R_{\textsf S}$. Under the above settings, the interference-correlation coefficients for the MCP model are the same with those for the SOCP model according to Theorem \ref{th:IntCorrCoe}. Thus, the results for the SOCP model are omitted.

\begin{figure}[t]
\centering
\subfigure[Correlation coefficients for different $c$, $\epsilon\!=\!0.01$, $R_{\textsf M}\!=\!5.$]{\includegraphics[width=8cm]{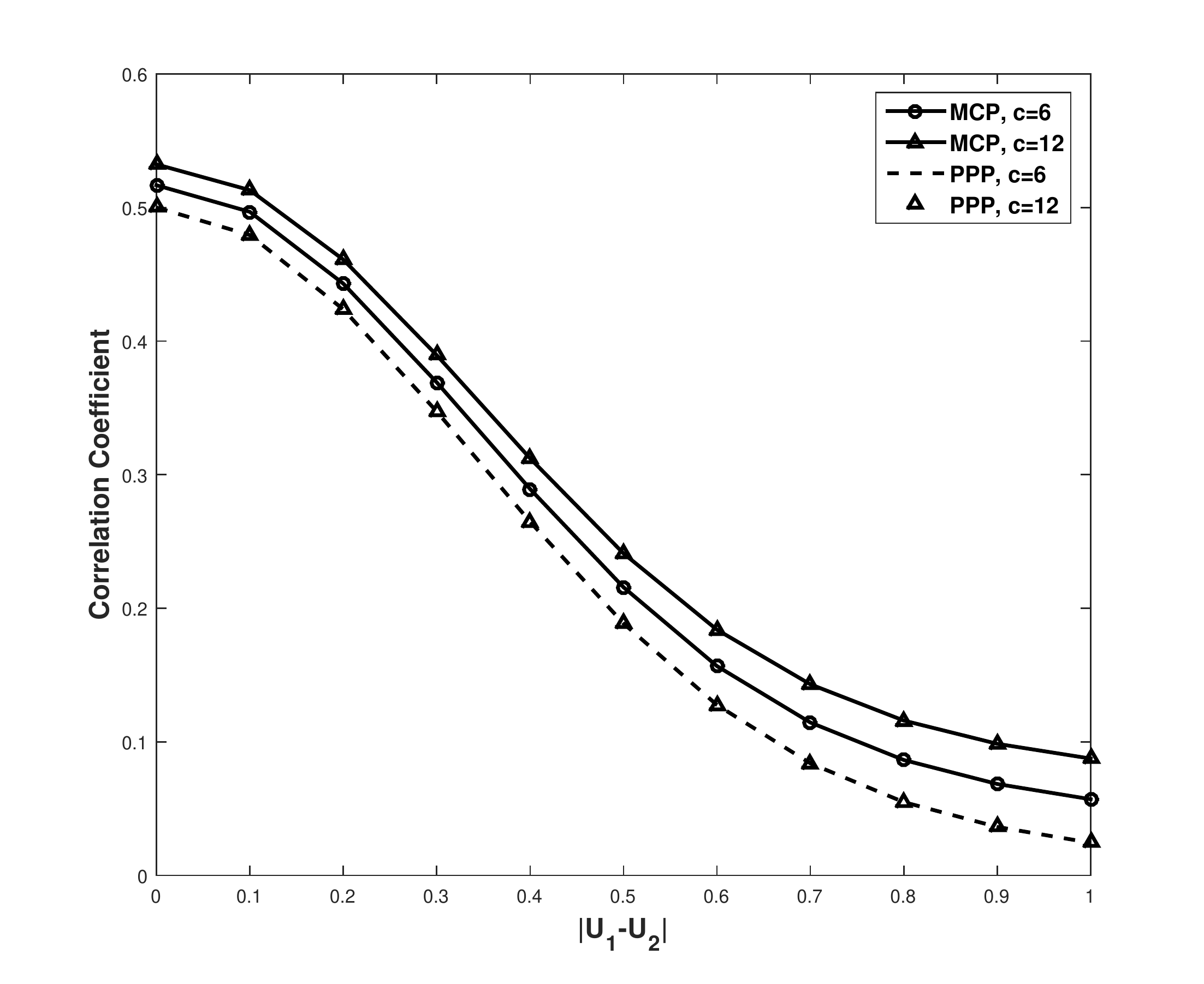}}
\subfigure[Correlation coefficients for different $R$, $\epsilon\!=\!0.001$, $c_{\textsf M}\!=\!3.$]{\includegraphics[width=7.5cm]{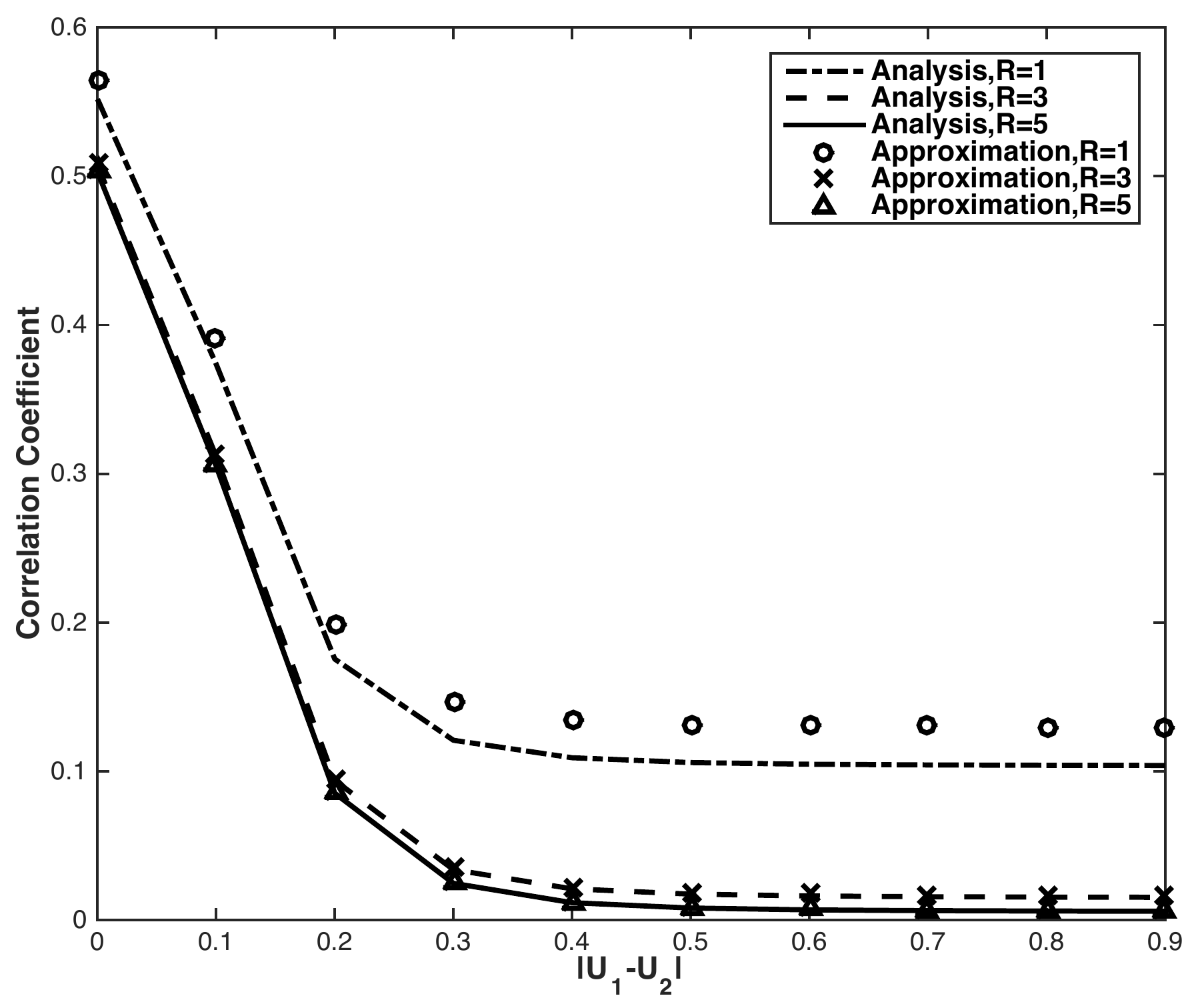}}
\caption{Interference-correlation coefficient versus the distance $|U_1-U_2|$. Here, $\alpha\!=\!4$, $P_s\!=\!43$ dBm, $\lambda_{\textsf M^o}\!=\!0.1$.}
\label{Fig:CorCoef_DifferCR}
\end{figure}

In Fig. \ref{Fig:CorCoef_DifferCR}, interference-correlation coefficients under different mean number of points in each cluster, $c$, and the cluster radius, $R$, are plotted in (a) and (b), respectively. The curves for the MCP model and the PPP model are computed numerically using Theorem \ref{th:IntCorrCoe} and Remark \ref{Rmk:CorCoef_PPP}, respectively. The approximation of interference correlation coefficient in Fig. \ref{Fig:CorCoef_DifferCR} (b) is calculated via substituting (\ref{Eq:Approx_F}) into (\ref{eq:CorrCoef_MCP}). According to system model, different $c$ indicates different interferer densities since $\lambda_{\textsf P}=\lambda_{\textsf M}=\lambda_{\textsf M^o} c_{\textsf M}$. First of all,  It is observed that the interference-correlation coefficients for the MCP model are greater than those for the PPP model. This suggests that BS clustering enlarges the interference correlation, which matches the conclusion in Proposition \ref{Prop:CorCoef_MCP_PPP}. Furthermore, it is found that the curves of correlation coefficients for the PPP model under different densities coincide with each other since $\zeta_{\textsf P}$ is independent of BS density according to (\ref{eq:CorrCoef_PPP}). In addition, also shown in the figures is that increasing $c$ or decreasing $R$ enlarges the interference correlations for the MCP model due to the increase in the attraction between the interfering BSs. Furthermore, Fig. \ref{Fig:CorCoef_DifferCR} (b) shows that the approximation of $F(c,R)$, i.e., Lemma \ref{lm:Approx_F}, evaluates the interference correlation coefficient well, even for the case of small $R$. 
\subsection{Joint Success Probability}

\begin{figure}[t]
\centering
\subfigure[Joint success probability for MUs]{\includegraphics[scale=0.35]{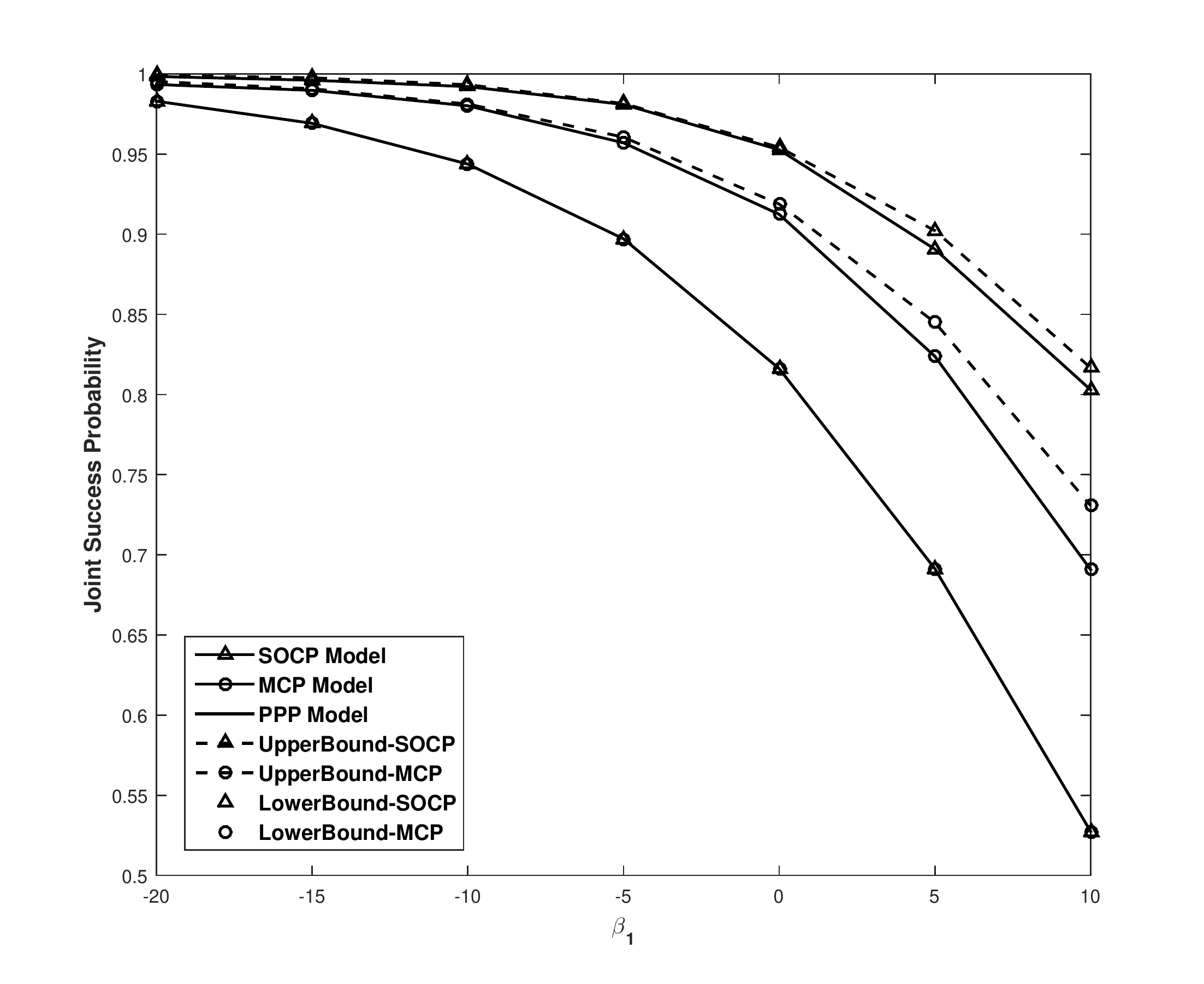}}
\subfigure[Joint success probability for SUs]{\includegraphics[scale=0.35]{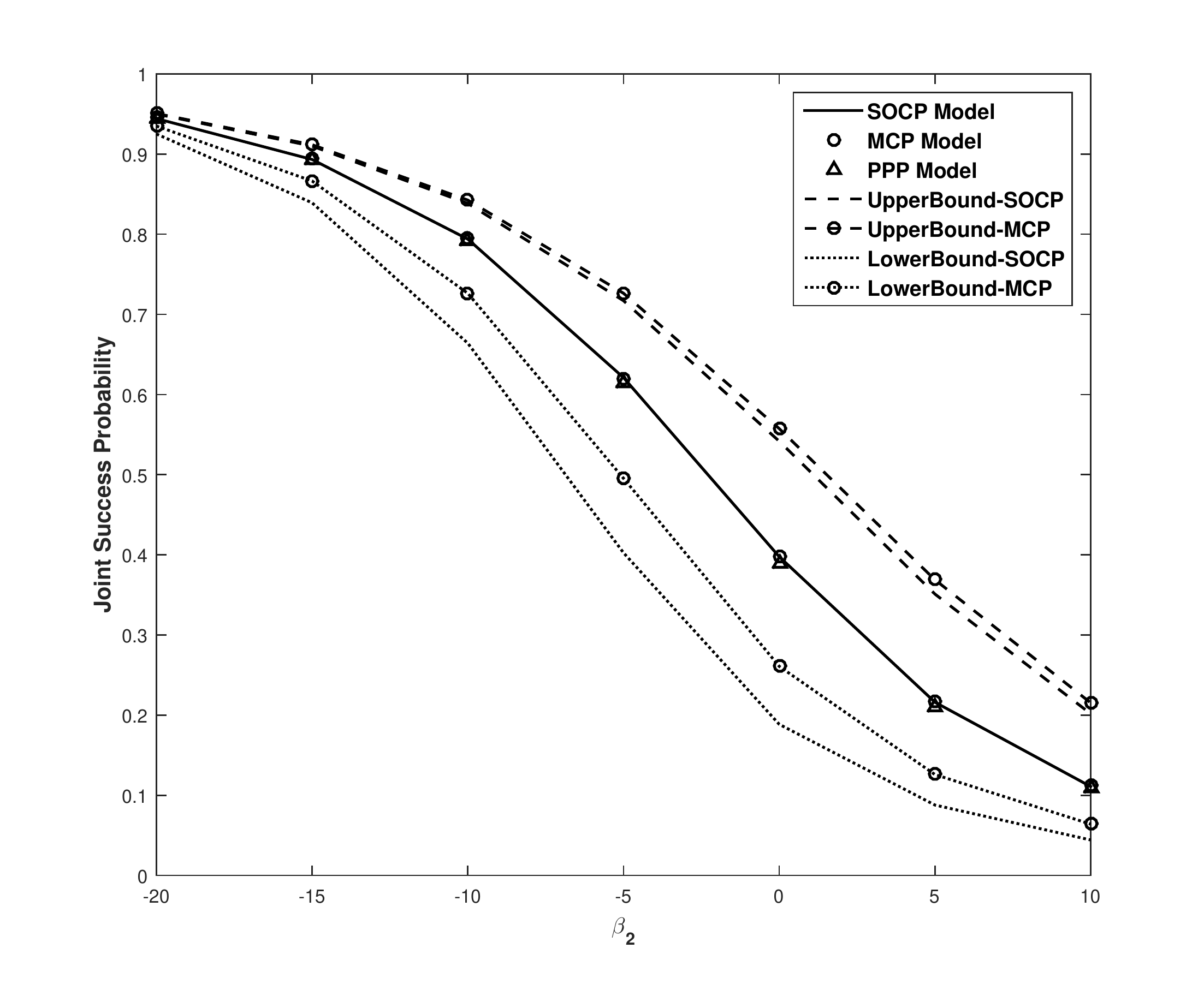}}
\centering\caption{Joint success probability versus SIR threshold. Here, $\alpha=4$, $\lambda_{0}=\lambda_{m}=7.96\times10^{-6}\:\mathrm{m^{-2}}$,
$\lambda_{p}=1.2\times10^{-4}\mathrm{m^{-2}}$, $P_{m}=39\:\mathrm{dBm}$,
$P_{s}=13\:\mathrm{dBm}$, $c_{1}=15$, $\sigma=50$, $c=c_{2}=3$} \label{Fig:JSP}
\end{figure}

Fig. \ref{Fig:JSP} shows the joint success probabilities and their corresponding bounds for MUs and SUs versus SIR threshold. The curves for the joint success probability for MUs and SUs are calculated by Theorem \ref{Thm:JSP} and the corresponding bounds by Theorem \ref{Thm:Bounds_JSP}. First of all, from Fig. \ref{Fig:JSP}(a), it is found that the curves of the lower bound of the joint success probability for MUs in the MCP model and SOCP model coincide with the joint success probability for the PPP model with the identical SBS density, which is verified by Theorem \ref{Thm:Bounds_JSP}. Next, it is also observed that the joint success probability for MUs increases in the order of the PPP, MCP, and SOCP models. The reason is that, there is a high probability for SBSs to be allocated at the edge of MBSs in SOCP model leading to a low inter-tier interference. Hence, in order to improve MU performance, it is suggested to deploy the SBSs in the annular region of MBSs. Last, Fig. \ref{Fig:JSP}(b) shows that there is little difference in the joint success probabilities for SUs in these three models. This is because, given the serving BS, the dominant interfering BSs comes from the same cluster, which is distributed as a PPP in the MCP and SOCP model.

\begin{figure}[t]
\centering \includegraphics[scale=0.3]{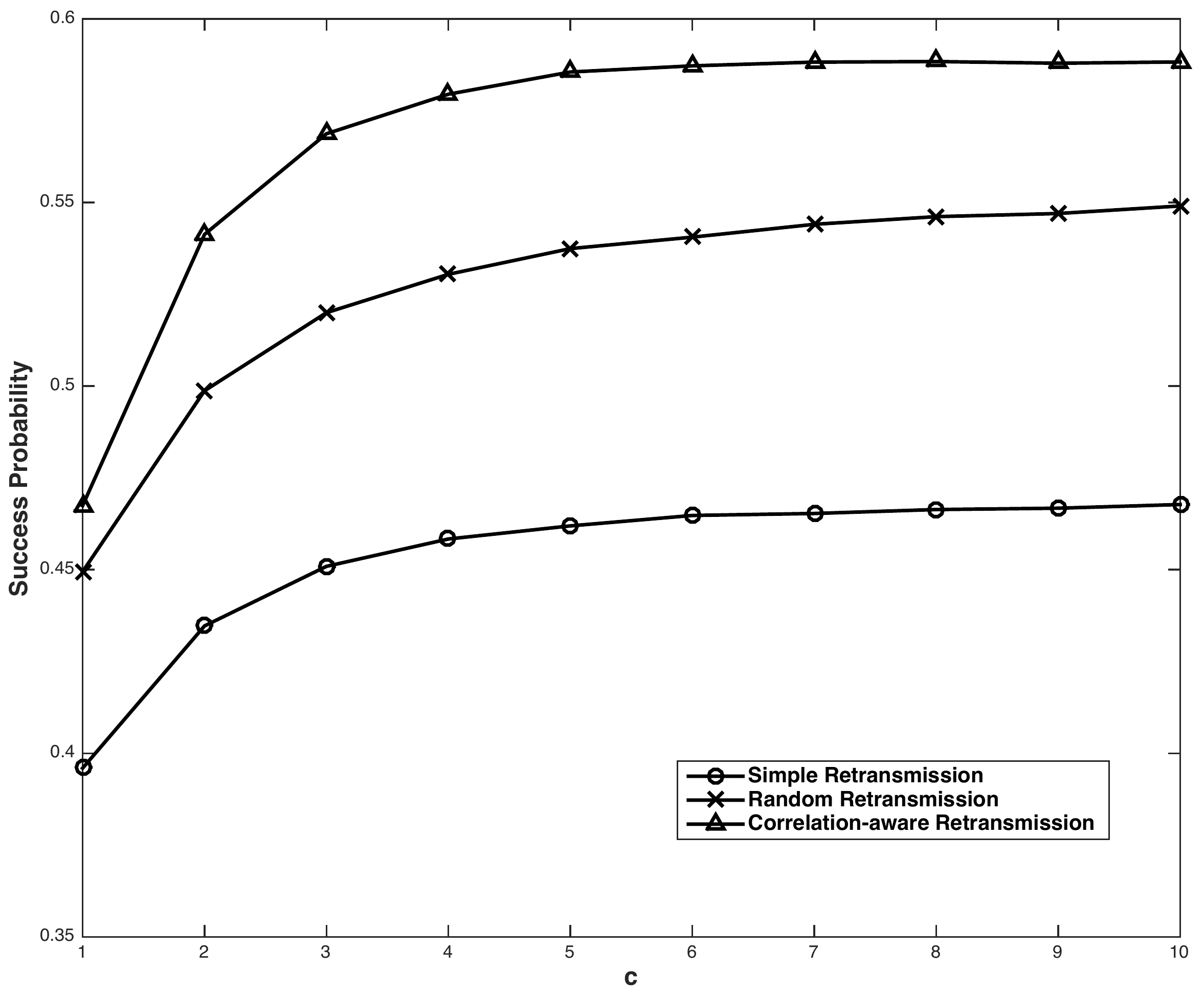}
\caption{Success probability for different retransmission schemes. Here, $\alpha=4$, $\lambda_{0}=\lambda_{m}=7.96\times10^{-6}\:\mathrm{m^{-2}}$,
$\lambda_{p}=1.2\times10^{-4}\mathrm{m^{-2}}$, $P_{m}=39\:\mathrm{dBm}$,
$P_{s}=13\:\mathrm{dBm}$, $\beta_1=-2$dB, $\beta_2=-3$dB.}\label{Fig:SuccProb_CompAlgor}
\end{figure}

\subsection{Correlation-aware Retransmission Scheme}

Fig. \ref{Fig:SuccProb_CompAlgor} compares the performance of HCNs under correlation-aware retransmission scheme proposed in this paper (Remark \ref{Rm:CorrRetrans}) with the simple and random retransmission schemes. For simple retransmission scheme, all BSs (re)transmit packets at all time slots. For random retransmission scheme \cite{LocalDelay_MAC}, all BSs (re)transmit packets with a given probability $p$. First of all, it is shown that the random retransmission scheme enhances the success probability since it introduces the randomness in transmission and thus reduces the interference correlation \cite{LocalDelay_MAC}. Next, correlation-aware retransmission scheme is observed to achieve higher success probability than the simple and random retransmission schemes since it takes advantage of the interference correlation when the success probability is high and avoid the blind retransmissions otherwise. Further, we observe that the gain increases with the growing $c$ (i.e. mean number of nodes in each cluster). This shows that effectively managing the effects of BS spatial interdependence on interference correlation significantly improve the network performance.

\section{Conclusions}

In this paper, we have studied the effects of BS spatial interdependence on interference correlation and the performance of HCNs with HARQ. While it is known that BS clustering degrades network performance, few results  exist on quantifying the effects of the phenomenon on interference correlation and closely related network performance with retransmissions. Our work makes contributions by analyzing such effects, revealing in a simple form how a growing level of clustering increases the interference-correlation coefficient. Specifically, it is shown that the interference-correlation coefficient is a monotone-increasing function of the mean number of nodes for each cluster and monotone-decreasing function of the cluster radius. Furthermore, we have presented a correlation-aware transmission scheme to illustrate how to take advantage of interference correlation and avoid the blind retransmissions for improving network performance. 

The used methodology and achieved results in this paper provide the way to quantify the effects of BS spatial interdependence on interference correlation and network performance. This work relies on cluster processes in stochastic geometry and some simplified assumptions to get the tradeoff between the mathematical tractability and practical network deployment. To derive more elaborate insight in practical networks with spatial dependence, further investigations in practical settings are necessary by considering other network deployment, using realistic channel model with correlation, and taking account of finite mobility of users in HCNs. Furthermore, studying the effects of BS interdependence on spatial interference correlation and the network performance under multi-hop transmissions is also an interesting topic.

\appendix

\subsection{Cluster Point Processes}\label{App:ClusterPP}
Two types of cluster point process, namely MCP and SOCP, are used for constructing the network model in Section~\ref{Section:System}. They are defined as follows.

Let a MCP be denoted as $\Phi_{\textsf M}$ with density $\lambda_{\textsf M}$. The process consists of a parent point process and a daughter point process forming clusters centered at different parent points. The parent point process is a PPP, denoted as $\Phi_{\textsf M^o}$,  with the  density $\lambda_{\textsf M^o}$. For a cluster, the daughter points
are uniformly distributed in a disk region with the  radius $R_{\textsf M}$ and  centered at the corresponding  parent point. The distance from a typical daughter point to the corresponding parent point has the following  probability density function (PDF):
\begin{equation}
f_{\textsf M}(r)=\begin{cases}
\frac{1}{\pi R_{\textsf M}^{2}}, & r\leq R_{\textsf M}, \\
0, & \text{otherwise}.
\end{cases}
\end{equation}
The number of daughter points in each cluster is  a Poisson-distributed random variable with mean $c_{\textsf M}$. Thus, the density of the MCP  is $\lambda_{\textsf M} = \lambda_{\textsf M^o} c_{\textsf M}$. Let $\mathcal{N}(X)$ denote a cluster centered at a parent point $X \in \Phi_{\textsf M^o}$. Then the MCP is given as $\Phi_{\textsf M} = \bigcup_{X\in \Phi_{\textsf M^o}} \mathcal{N}(X)$. The distribution of the MCP is illustrated in Fig. \ref{Fig:Network}(a).

Next, let a SOCP be denoted as $\Phi_{\textsf S}$ with density $\lambda_{\textsf S}$. The process consists of a parent point process, a first-order cluster process, and a daughter point process. The parent point process is a PPP, denoted as $\Phi_{\textsf S^o}$, with the density $\lambda_{\textsf S^{o}}$. For the first-order cluster, the points are isotropically scattered in a disk region with the radius $R_{\textsf S^\prime}$ and centered at the corresponding parent point. The distance from a typical first-order cluster point to the corresponding parent point has the following reverse Gaussian distribution \cite{HCN_SOCP}:
\begin{equation}
f_{\textsf S^\prime}(r)=\begin{cases}
\frac{\left(1-\exp\left(\frac{-r^{2}}{2\sigma^{2}}\right)\right)}{\pi R_{\textsf S^\prime}^{2}+2\pi\sigma^{2}\left(\exp\left(\frac{-R_{\textsf S^\prime}^{2}}{2\sigma^{2}}\right)-1\right)}, & r\leq R_{\textsf S^\prime}\\
0, & \text{otherwise},
\end{cases}
\end{equation}
where $\sigma$ denotes the standard deviation of reverse Gaussian distribution. Furthermore, for the second-order cluster, the daughter points are uniformly distributed in a disk region with the radius $R_{\textsf S}$ and centered at the corresponding first-order cluster point. The distance from a typical daughter point to the corresponding center (first-order cluster point) has the following PDF:
\begin{equation}
f_{\textsf S}(r)=\begin{cases}
\frac{1}{\pi R_{\textsf S}^{2}}, & r\leq R_{\textsf S}\\
0, & \text{otherwise}.
\end{cases}
\end{equation}
The number of points in each first-order cluster and second-order cluster are Poisson-distributed random variable with mean $c_{\textsf S^\prime}$ and $c_{\textsf S}$, respectively. Thus, the density of the SOCP is $\lambda_{\textsf S} = \lambda_{\textsf S^o} c_{\textsf S^\prime} c_{\textsf S}$. Let $\mathcal{N}(X^{[Y]})$ denote a cluster centered at a first-order cluster point $X^{[Y]} \in \Phi_{\textsf S^\prime}$ with the parent point $Y \in \Phi_{\textsf S^o}$. Then the SOCP is given as $\Phi_s = \bigcup_{X^{[Y]}\in \Phi_{\textsf S^\prime}} \bigcup_{Y\in \Phi_{\textsf S^o}} \mathcal{N}(X^{[Y]})$. The distribution of the SOCP is illustrated in Fig.~\ref{Fig:Network}(b).

\subsection{Radius of the association area for MBSs and SBSs}\label{RadiusOfAssociationArea}
Consider the  MCP model, the average coverage area of each cluster of SBSs or each MBS is $(\lambda_m + \lambda_{\textsf M^o})^{-1}$ where $\lambda_m$ and $\lambda_{\textsf M^o}$ are the densities of the MBSs and the parent process of the SBSs, respectively. Hence  the corresponding coverage radius is $D_{m \textsf M} = \left[\sqrt{\pi (\lambda_m + \lambda_{\textsf M^o})}\right]^{-1}$. Since there are $c_{\textsf M}$ SBSs in each cluster on  average, the average coverage area of each SBS is $[c (\lambda_m + \lambda_{\textsf M^o})]^{-1}$ and hence the corresponding coverage radius is $D_{s \textsf M}=D_{m \textsf M} / \sqrt{c_{\textsf M}}$. Next, consider the SOCP model.  The average coverage area of each MBS or  each cluster of first-order points is  $(\lambda_m + c_{\textsf S^\prime} \lambda_{\textsf S^o})^{-1}$. This results in  the average coverage radius being   $D_{m \textsf S} =\left[\sqrt{\pi (\lambda_m + c_{\textsf S^\prime} \lambda_{\textsf S^o})}\right]^{-1}$. Since the average number of first-order points in each cluster is $c_{\textsf S^\prime}$, the average coverage area of each cluster of SBSs is $\left[c_{\textsf S^\prime} (\lambda_m + c_{\textsf S^\prime} \lambda_{\textsf S^o})\right]$. Furthermore, there are $c_{\textsf S}$ SBSs in each cluster on average. Then the average coverage area of each SBS is $\left[c_{\textsf S} c_{\textsf S^\prime} (\lambda_m + c_{\textsf S^\prime} \lambda_{\textsf S^o})\right]$ and its coverage radius is $D_{s \textsf S}=D_{m \textsf S} / \sqrt{c_{\textsf S^\prime} c_{\textsf S}}$.

\subsection{Proof of Lemma \ref{lm:MeanInterference}}\label{pf:MeanInterference}
Here, we only show the main steps for the mean interference in the MCP model and omit those for the SOCP model since they follow the similar steps.

The interference power measured at the location $U$ in time slot $t$ is given as
\begin{equation}
I_{\epsilon}(U,t) =\sum_{X\in\Phi_{\textsf M}}h_{XU}(t)g_{\epsilon}(X-U)\nonumber =\sum_{Z\in\Phi_{\textsf M^o}}\sum_{X\in\Phi_{\textsf M}^{[Z]}}h_{XU}(t)g_{\epsilon}(X-U),
\end{equation}
where $\Phi_{\textsf M}^{[Z]}$ denotes the cluster associated with parent point $Z\in\Phi_{\textsf M^o}$.

The mean interference is given by
\begin{flalign}
& \mathbb{E}[I_{\epsilon}(U,t)]\nonumber  =\mathbb{E}[\sum_{Z \in \Phi_{\textsf M^o}}\sum_{X \in \Phi_{\textsf M}^{[Z]}}h_{XU}(t)g_{\epsilon}(X-U)] \overset{(a)}{=}\mathbb{E}[h] \lambda_{\textsf M^o} \int_{\mathbb{R}^{2}}\mathbb{E}[\sum_{X \in \Phi_{\textsf M}^{[Z]}}g_{\epsilon}(X-U)]\mathrm{d}Z\nonumber \\
 & \overset{(b)}{=}\mathbb{E}[h]\lambda_{\textsf M^o} c_{\textsf M} \!\!\int_{\mathbb{R}^{2}} \!\int_{\mathbb{R}^{2}} \!g_{\epsilon}(X\!-\!U\!-\!Z)f_{\textsf M}(X) \mathrm{d}X\mathrm{d}Z\nonumber \!=\! P_s \mathbb{E}[h] \lambda_{\textsf M^o} c_{\textsf M} \!\!\int_{\mathbb{R}^{2}} \!g_{\epsilon}(X\!-\!U) \!\!\int_{\mathbb{R}^{2}}f_{\textsf M}(X\!+\!Z)\mathrm{d}Z\mathrm{d}X\nonumber \\
 & \overset{(c)}{=}\mathbb{E}[h] \lambda_{\textsf M^o} c_{\textsf M} \int_{\mathbb{R}^{2}}g_{\epsilon}(X)\mathrm{d}X,\label{eq:MeanInterf_MCP}
\end{flalign}
where $(a)$ and $(b)$ come from Campbell-Mecke Theorem, $(c)$ follows from $\int_{\mathbb{R}^{2}}f_{\textsf M}(X)\mathrm{d}X=1$.

\subsection{Proof of Lemma \ref{lm:Approx_F}}\label{pf:Approx_F}
Given $R$ is large, $F(c,R)$ can be approximated as
\begin{flalign}
F(c,R) 
& = \frac{c}{\pi^{2}R^{4}}\int_{\mathbb{R}^{2}}\int_{\mathbb{R}^{2}}g_{\epsilon}(X)g_{\epsilon}(Y)A_{R}(|X-Y|)\mathrm{d}X\mathrm{d}Y \nonumber \\
& \overset{(a)} {=} \frac{c}{\pi R^{2}}\int_{\mathbb{R}^{2}}\int_{\mathbb{R}^{2}}g_{\epsilon}(X)g_{\epsilon}(Y) \mathrm{d}X\mathrm{d}Y + o (1/R^2)\nonumber \\
& = \frac{c}{\pi R^{2}} \left[\int_{\mathbb{R}^{2}} g_{\epsilon}(X) \mathrm{d}X \right]^2 + o (1/R^2) 
   = \frac{c}{\pi R^{2}} \left[ \int_{0}^{2 \pi} \int_{0}^{\infty} \frac{r}{r^{\alpha} + \epsilon} \mathrm{d}r \right]^2 + o (1/R^2)\nonumber \\
& \overset{(b)} {=}\frac{4 c \pi^3 }{\alpha^2 R^2} \epsilon^{{4-2 \alpha}/\alpha} (\csc({2 \pi}/\alpha))^2 + o (1/R^2)
\end{flalign}
where $(a)$ comes from the fact that $A_{R}(|X-Y|) \approx \pi R^2$, for large $R$ and $(b)$ uses the formula $\int_{0}^{\infty} \frac{r}{r^{\alpha} + \epsilon} \mathrm{d}r = \alpha^{-1} \epsilon^{{2-\alpha}/\alpha} \pi \csc({2 \pi}/\alpha)$ given in \cite[Eqn. 3.241.4]{IntegralTable}

\subsection{Proof of Lemma \ref{lm:SecondMoment}}\label{pf:SecondMoment}
1) MCP model:
\begin{flalign}
& \mathbb{E}[I_{\epsilon}(U_1,t_1),I_{\epsilon}(U_2,t_2)]\nonumber =\mathbb{E}\left[\sum_{X\in\Phi_{\textsf M}}\!\!h_{XU_1}(t_{1})g_{\epsilon}(X-U_1)\sum_{Y\in\Phi_{\textsf M}}\!\!h_{YU_2}(t_{2}) g_{\epsilon}(Y-U_2)\right]\nonumber \\
 =& \underbrace{ \mathbb{E}\!\left[\sum_{X\in\Phi_{\textsf M}}\!\!h_{XU_1}(t_{1})h_{XU_2}(t_{2})g_{\epsilon}(X\!\!-\!\!U_1)g_{\epsilon}(X\!\!-\!\!U_2)\right]}_{\xi_1}\!+\!\underbrace{\mathbb{E}\!\left[\sum_{x,y\in\Phi_{\textsf M}}^{X\neq Y}\!\!h_{XU_1}(t_{1})h_{YU_2}(t_{2})g_{\epsilon}(X\!\!-\!\!U_1)g_{\epsilon}(Y\!\!-\!\!U_2)\right]}_{\xi_2}.\label{eq:Mean product_original}
\end{flalign}

Next, we calculate $\xi_1$ and $\xi_2$, respectively.
\begin{equation}
 \xi_1\!=\!\mathbb{E}[h]^{2}\mathbb{E}\!\!\left[\sum_{Z\in\Phi_{\textsf M^o}}\!\sum_{X\in\Phi_{\textsf M}^{[Z]}}g_{\epsilon}(X\!-\!U_1)g_{\epsilon}(X\!-\!U_2)\!\right]\!\!\overset{(a)}{=}\mathbb{E}[h]^{2}\lambda_{\textsf M^o} c_{\textsf M} \int_{\mathbb{R}^{2}}g_{\epsilon}(X\!-\!U_1)g_{\epsilon}(X\!-\!U_2)\mathrm{d}X,\label{eq:F}
\end{equation}
where $(a)$ comes from Campbell-Mecke Theorem and the
fact that $\int_{\mathbb{R}^{2}}f_{\textsf M}(X)\mathrm{d}X=1$.
\begin{flalign}
\xi_2 & =\mathbb{E}[h]^{2}\mathbb{E}\left[\sum_{X,Y\in\Phi_{s}}^{X\neq Y}g_{\epsilon}(X-U_1)g_{\epsilon}(Y-U_2)\right]
  \overset{(a)}{=}\mathbb{E}[h]^{2}\int_{\mathbb{R}^{2}}\int_{\mathbb{R}^{2}}g_{\epsilon}(X)g_{\epsilon}(Y)\rho_{\textsf M}^{(2)}(X,Y)\mathrm{d}X\mathrm{d}Y\nonumber \\
 & \overset{(b)}{=}\mathbb{E}[h]^{2}\left(\lambda_{\textsf M^o}c_{\textsf M} \right)^{2}\int_{\mathbb{R}^{2}}\int_{\mathbb{R}^{2}}g_{\epsilon}(X)g_{\epsilon}(Y)\mathrm{d}X\mathrm{d}Y
  +\mathbb{E}[h]^{2}\lambda_{\textsf M^o}c_{\textsf M} F(c_{\textsf M},R_{\textsf M}),\label{eq:Q}
\end{flalign}
where $(a)$ follows from that $X$ can be substituted by $X-U_1$ and $Y$ can be substituted by $Y-U_2$ in the integrals, $(b)$ comes from the second moment density of a MCP given by \cite[p. 128]{StoGeoBook-Martin}, $F(\cdot,\cdot)$ is given in (\ref{Eq:Fun1}). Substituting (\ref{eq:F}) and (\ref{eq:Q}) into (\ref{eq:Mean product_original}), we get the mean product of $I_{\epsilon}(U_1,t_1)$ and $I_{\epsilon}(U_2,t_2)$ in Lemma \ref{lm:SecondMoment}.

Based on the results of $\mathbb{E}[I_{\epsilon}(U_1,t_1),I_{\epsilon}(U_2,t_2)]$, the second moment of interference is given by
\begin{flalign}
& \mathbb{E}[I_{\epsilon}^2(U,t)]\nonumber =\mathbb{E}\left[\sum_{X\in\Phi_{\textsf M}}\!\!h_{XU}(t)g_{\epsilon}(X-U)\sum_{Y\in\Phi_{\textsf M}}\!\!h_{YU}(t) g_{\epsilon}(Y-U)\right]\nonumber \\
 =& \underbrace{ \mathbb{E}\!\left[\sum_{X\in\Phi_{\textsf M}}\!\!h_{XU}^2(t)g_{\epsilon}^2(X\!-\!U)\right]}_{\xi_1^\prime}\!+\!\underbrace{\mathbb{E}\!\left[\sum_{X,Y\in\Phi_{\textsf M}}^{X\neq Y}\!\!h_{XU}(t)h_{YU}(t)g_{\epsilon}(X\!-\!U)g_{\epsilon}(Y\!-\!U)\right]}_{\xi_2^\prime} \nonumber \\
 =& \underbrace{\mathbb{E}[h^2]\lambda_{\textsf M^o} c_{\textsf M} \! \int_{\mathbb{R}^{2}}\!g_{\epsilon}^2(X)\mathrm{d}X}_{\xi_1^\prime} \!+ \! \underbrace{\mathbb{E}[h]^{2}\lambda_{\textsf M^o}c_{\textsf M} \!\left[\lambda_{\textsf M^o}c_{\textsf M} \!\int_{\mathbb{R}^{2}}\!\int_{\mathbb{R}^{2}}\!g_{\epsilon}(X)g_{\epsilon}(Y)\mathrm{d}X\mathrm{d}Y
  \!+ \!F(c_{\textsf M},R_{\textsf M})\right]}_{\xi_2^\prime}
\end{flalign}

2) SOCP model:

Following the similar steps, we derive the $\mathbb{E}[I_{\epsilon}(U_1,t_1),I_{\epsilon}(U_2,t_2)]$ and $\mathbb{E}[I_{\epsilon}^2(U,t)]$ in the SOCP model as:
\begin{equation}
\mathbb{E}[I_{\epsilon}(U_1,\!t_1),I_{\epsilon}(U_2,\!t_2)] \!=\! \mathbb{E}[h]^{2}\lambda_{\textsf S} \!\int_{\mathbb{R}^{2}}\!g_{\epsilon}(X\!-\!U_1)g_{\epsilon}(X\!-\!U_2)\mathrm{d}X \!+\! \mathbb{E}[h]^{2}\!\int_{\mathbb{R}^{2}}\!\!\int_{\mathbb{R}^{2}}\!g_{\epsilon}(X)g_{\epsilon}(Y)\rho_{\textsf S}^{(2)}(X,\!Y)\mathrm{d}X\mathrm{d}Y \label{eq:Mean_product_SOCP}
\end{equation}
\begin{equation}
\mathbb{E}[I_{\epsilon}^2(U,t)] = \mathbb{E}[h]^{2}\lambda_{\textsf S} \int_{\mathbb{R}^{2}}g_{\epsilon}^2(X)\mathrm{d}X + \mathbb{E}[h]^{2}\int_{\mathbb{R}^{2}}\int_{\mathbb{R}^{2}}g_{\epsilon}(X)g_{\epsilon}(Y)\rho_{\textsf S}^{(2)}(X,Y)\mathrm{d}X\mathrm{d}Y \label{eq:SecondMoment_SOCP}
\end{equation}
where $\rho_{\textsf S}^{(2)}$ denotes the second moment density of the SOCP.

The key of calculating (\ref{eq:Mean_product_SOCP}) and (\ref{eq:SecondMoment_SOCP}) is to derive $\rho_{\textsf S}^{(2)}$. According to \cite[pp127]{StoGeoBook-Martin}, the second moment density of SOCP is expressed as
\begin{equation}
\rho_{\textsf S}^{(2)}(X,Y)=\lambda_{\textsf S}^{2}+\underbrace{\mathbb{E}\left[\sum_{Z_0\in\Phi_{\textsf S^o}}\sum_{Z_1\in\Phi_{\textsf S^\prime}^{[Z_0]}}\rho(X,Y\mid Z_0,Z_1)\right]}_{\rho_{\textsf S}^\prime}, \label{Eq:SMD_SOCP}
\end{equation}
where $\rho(X,Y\mid Z_0,Z_1)$ denotes the conditional second moment density given the parent point $Z_0 \in \Phi_{\textsf S^o}$ and the first cluster point $Z_1 \in \Phi_{\textsf S^\prime}^{[z_0]}$.

Next, we calculate $\rho_{\textsf S}^\prime$ to derive $\rho_{\textsf S}^{(2)}(X,Y)$.
\begin{flalign}
 \rho_{\textsf S}^\prime
\overset{\left(a\right)}{=}& \mathbb{E}\left[\sum_{Z_0\in\Phi_{\textsf S^o}} \sum_{Z_1\in\Phi_{\textsf S^\prime}^{[Z_0]}} c_{\textsf S}f_{\textsf S}(X-Z_1-Z_0)c_{\textsf S}f_{\textsf S}(Y-Z_1-Z_0)\right]\nonumber \\
\overset{\left(b\right)}{=} & \lambda_{\textsf S^o}c_{\textsf S^\prime}c_{\textsf S}^2 \int_{\mathbb{R}^{2}}\!\int_{\mathbb{R}^{2}}f_{\textsf S^\prime}(Z_1)f_{\textsf S}(X-Z_1-Z_0)f_{\textsf S}(Y-Z_1-Z_0)\mathrm{d}Z_0 \mathrm{d}Z_1\nonumber \\
\overset{\left(c\right)}{=} & \lambda_{\textsf S^o}c_{\textsf S^\prime}c_{\textsf S}^2 (f_{\textsf S}\star f_{\textsf S})(X-Y)\int_{\mathbb{R}^{2}}f_{\textsf S^\prime}(Z_1)\mathrm{d}Z_1\nonumber \\
\overset{\left(d\right)}{=} & \lambda_{\textsf S^o}c_{\textsf S^\prime}c_{\textsf S}^2 (f_{\textsf S}\star f_{\textsf S})(X-Y)\nonumber \\
\overset{\left(e\right)}{=} & \lambda_{\textsf S^o}c_{\textsf S^\prime}c_{\textsf S}^2\cdot\frac{A_{R_{\textsf S}}(| X-Y|)}{\pi^{2}R_{\textsf S}^{4}},\label{Eq:CSMD_SOCP}
\end{flalign}
where $\left(a\right)$ follows from the independence of the points
in the same cluster, $\left(b\right)$ comes from Campbell-Mecke Theorem, $\left(c\right)$ comes from the definition of convolution $\star$,
$\left(d\right)$ follows from the fact that $\int_{\mathbb{R}^{2}}f_{\textsf S^\prime}\left(Z_1\right)\mathrm{d}Z_1=1$,
$\left(e\right)$ comes from the calculation of $\left(f_{\textsf S}\star f_{\textsf S}\right)\left(X-Y\right)$
which is given in \cite{StoGeoBook-Martin} and $A_{R_{\textsf S}}(r)=2R_{\textsf S}^{2}\arccos\left(\frac{r}{2R_{\textsf S}}\right)-r\sqrt{R_{\textsf S}^{2}-\frac{r^{2}}{4}},\;0\leq r\leq2R_{\textsf S}$.

Last, the mean product (or the second moment) of the interference power are derived by substituting (\ref{Eq:SMD_SOCP}) and (\ref{Eq:CSMD_SOCP}) into (\ref{eq:Mean_product_SOCP}) (or (\ref{eq:SecondMoment_SOCP})).

\subsection{Proof of Proposition \ref{Prop:CorCoef_MCP_PPP}}\label{pf:CC_MCP_PPP}
To notational simplicity, let $\theta=\int_{\mathbb{R}^{2}}g(X)g(X-\|U_1-U_2\|)\mathrm{d}X>0$, and
$\theta^{\prime}=\frac{\mathbb{E}\left[h^{2}\right]}{\mathbb{E}\left[h\right]^{2}}\int_{\mathbb{R}^{2}}g^{2}(X)\mathrm{d}X>0$.
Hence, $\zeta_{\textsf P}$ is expressed as $\zeta_{\textsf P}=\frac{\theta}{\theta^{\prime}}$, and both $\zeta_{\textsf M}$ and $\zeta_{\textsf S}$ can be written as $\zeta=\frac{\theta+F\left(c,R\right)}{\theta^{\prime}+F\left(c,R\right)}$, where $(c,R)=(c_{\textsf M},R_{\textsf M})$ for $\zeta_{\textsf M}$ and $(c,R)=(c_{\textsf S}, R_{\textsf S}$) for $\zeta_{\textsf S}$.

Next, we show that $\zeta \geq \zeta_{\textsf P}$.
\begin{equation}
\zeta-\zeta_{\textsf P}=\frac{\theta+F(c,R)}{\theta^{\prime}+F(c,R)}-\frac{\theta}{\theta^{\prime}}=\frac{(\theta^{\prime}-\theta)F(c,R)}{\theta^{\prime}(\theta^{\prime}+F(c,R))}\overset{\left(a\right)}{\geq 0},\label{Eq:Comparison}
\end{equation}
where $(a)$ comes from the fact that $F(c,R)\geq 0$ and $\theta^{\prime}-\theta>0$
since $0<\zeta_{\textsf P}=\frac{\theta}{\theta^{\prime}}<1$ \cite{InterferenceCorreLetter}. The equality of (\ref{Eq:Comparison}) holds when $F(c,R)=0$.

\subsection{Proof of Proposition \ref{Prop:CorCoef_c_R}}\label{pf:IntCorCoef_c_R}

Proposition \ref{Prop:CorCoef_c_R} is proved by the following two steps. First, we show that $\zeta$ is a monotone-increasing function of $F(\cdot,\cdot)$. To this end, we take the derivative of $\zeta$ with respect to $F(\cdot,\cdot)$ and get that $\zeta'=\frac{\theta^{\prime}-\theta}{(\theta^{\prime}+\theta)^{2}}>0$. Therefore, $\zeta$ increases with the increase in $F(\cdot,\cdot)$.

Next, the function $F(c,R)$ is proved to be a monotone-increasing function of $c$ and monotone-decreasing function of $R$. Recall that $F(c,R)=\frac{c}{\pi^{2}R^{4}}\int_{\mathbb{R}^{2}}\int_{\mathbb{R}^{2}}g(X)g(Y)A_{R}(|X-Y|)\mathrm{d}X\mathrm{d}Y$,
where $A_{R}(r)=2R^{2}\mathcal{\arccos}\left(\frac{r}{2R}\right)-r\sqrt{R^{2}-\frac{r^{2}}{4}},\;0\leq r\leq2R$,
and 0 \emph{for} $r>2R$. From the expression of $F$, we get that $F(c,R)\propto c$
and $F(c,R)$ changes on the order of $k(\frac{1}{R})\cdot\frac{1}{R^2}$, where $k(\frac{1}{R})\in(0,\pi)$.

In particular, $F(c,R) \rightarrow 0$, if $\frac{c}{\pi^{2}R^{2}} \rightarrow 0$.
According to Proposition \ref{Prop:CorCoef_MCP_PPP}, $\zeta \rightarrow \zeta_{\textsf P}$, if $\frac{c}{\pi^{2}R^{2}} \rightarrow 0$.

\subsection{Proof of Lemma \ref{lm:ConJointSuccProb}}\label{pf:Cond_JSP}

According to the definition of the joint success probability, we have
\begin{flalign}
& \mathcal{P}_{m}^{\left(n\right)}(r_m) \nonumber\\& =\mathbb{P}\left(\frac{P_{m}h_{X_{m}}(t_{1})r_m^{-\alpha}} {I_{mm}(t_{1})+I_{sm}(t_{1})}>\beta_{m},\cdots,\frac{P_{m}h_{X_{m}}(t_{n})r_m^{-\alpha}}{I_{mm}(t_{n})+I_{sm}(t_{n})}>\beta_{m}\right)\nonumber \\
 & =\mathbb{P}\left(h_{X_{m}}(t_{1})>\frac{\beta_{m}\left(I_{mm}(t_{1})+I_{sm}(t_{1})\right)}{P_{m}r_m^{-\alpha}},\cdots,h_{X_m}\left(t_{n}\right)>\frac{\beta_{m}\left(I_{mm}(t_{n})+I_{sm}(t_{n})\right)}{P_{m}r_m^{-\alpha}}\right)\nonumber \\
 & \overset{\left(a\right)}{=}\mathbb{E}\left(\exp\left[\frac{-\beta_{m}\left(I_{mm}(t_{1})+I_{sm}(t_{1})\right)}{P_{m}r_m^{-\alpha}}\right]\times\cdots\times\exp\left[\frac{-\beta_{m}\left(I_{mm}(t_{n})+I_{sm}(t_{n})\right)}{P_{m}r_m^{-\alpha}}\right]\right)\nonumber \\
 & \overset{\left(b\right)}{=}\mathbb{E}_{\Phi_{m},\!\Phi_{s}}\left\{ \!\prod_{X\in\Phi_{m},\!X\neq X_{m}} \!\!\!\!\!\!\!\mathbb{E}_{h}\!\left[\!\exp\left(\!\frac{-\beta_{m} \widetilde{g}(X,r_m)}{r_m^{-\alpha}}\!\sum_{i=1}^{n}h_{X}(t_{i})\!\right)\!\right]\!\prod_{Y\in\Phi_{s}}\!\!\mathbb{E}_{h}\!\!\left[\exp\left(\!\frac{-\beta_{m}P_{s}|Y|^{-\alpha}}{P_{m}r_m^{-\alpha}}\sum_{i=1}^{n}h_{Y}(t_{i})\!\right)\!\right]\!\right\} \nonumber \\
 & \overset{\left(c\right)}{=}\mathbb{E}_{\Phi_{m}}\left\{ \prod_{X\in\Phi_{m},X\neq X_{m}}\mathbb{\mathbb{E}}_{h}\left[\exp\left(\frac{-\beta_{m}\widetilde{g}(X,r_m)}{r_m^{-\alpha}}\sum_{i=1}^{n}h_{X}(t_{i})\right)\right]\right\} \nonumber \\
 & \times\mathbb{E}_{\Phi_{s}}\left\{ \prod_{Y\in\Phi_{s}}\mathbb{\mathbb{E}}_{h}\left[\exp\left(\frac{-\beta_{m}P_{s}|Y|^{-\alpha}}{P_{m} r_m^{-\alpha}}\sum_{i=1}^{n}h_{Y}(t_{i})\right)\right]\right\} \nonumber \\
 & \overset{\left(d\right)}{=}\mathbb{E}_{\Phi_{m}}\left[\prod_{X\in\Phi_{m},X\neq X_{m}}\left(1+\frac{\beta_{m}\widetilde{g}(X,r_m)}{r_m^{-\alpha}}\right)^{-n}\right]\times\mathbb{E}_{\Phi_{s}}\left[\prod_{Y\in\Phi_{s}}\left(1+\frac{\beta_{m}P_{s}|Y|^{-\alpha}}{P_{m}r_m^{-\alpha}}\right)^{-n}\right]\nonumber \\
 & \overset{\left(e\right)}{=}G_{\Phi_{m}^{!}}\left[\left(1+\frac{\beta_{m}\widetilde{g}(X,r_m)}{r_m^{-\alpha}}\right)^{-n}\right]G_{\Phi_{s}}\left[\left(1+\frac{\beta_{m}P_{s}|Y|^{-\alpha}}{P_{m}r_m^{-\alpha}}\right)^{-n}\right],\label{eq:}
\end{flalign}
where $\left(a\right)$ comes from the independence of Rayleigh fading
channels, $\left(b\right)$ follows from the expression of $I_{mm}$
and $I_{sm}$ and $\widetilde{g}(X,r_m)=|X|^{-\alpha}\mathds{1}(|X|>r_m)$, $\left(c\right)$ comes from the fact that $\mathbb{E}_{\Phi_{m},\!\Phi_{s}}\!\left[A\!\left(\Phi_{m}\right)\!B\!\left(\Phi_{s}\right)\right]\!=\!\mathbb{E}_{\Phi_{m}}\!\left\{ \mathbb{E}_{\Phi_{s}}\!\left[A\left(\Phi_{m}\right)\!B\!\left(\Phi_{s}\right)\right]\right\} =\mathbb{E}_{\Phi_{m}}\left[A\left(\Phi_{m}\right)\right]\mathbb{E}_{\Phi_{s}}\left[B\left(\Phi_{s}\right)\right]$,
$\left(d\right)$ follows from the independence of Rayleigh fading channels, $\left(e\right)$comes from the definition of the PGF of point processes.

Similarly, we get the joint success probability for the typical SU as shown in Lemma 1.

\subsection{Proof of Lemma \ref{lm:bounds_PGF_SOCP}}\label{pf:Bouds_socp}
1) The lower bound of $G_{\Phi_{\textsf S}}(\lambda_{\textsf S})$:

According to (\ref{eq:PGF_SOCP}), the PGF of SOCP is expressed as
\begin{flalign}
G_{\Phi_{\textsf S}}(\lambda_{\textsf S})
&=\exp \!\!\left\{-\lambda_{\textsf S^o}\!\!\int_{\mathbb{R}^{2}}\!\!\left[ 1-M_{1}\!\!\left(\underbrace{\int_{\mathbb{R}^{2}}M_{2}(\int_{\mathbb{R}^{2}}v(X\!+\!Y+\!Z)f_{\textsf S}(Z)\mathrm{d}Z)f_{\textsf S^\prime}(Y)\mathrm{d}Y}_{T_0}\!\right)\right] \mathrm{d}X\!\right\}\nonumber\\
 & =\exp\left\{ -\lambda_{\textsf S^o}\int_{\mathbb{R}^{2}}\left[1-\exp\left(-c_{\textsf S^\prime}\left(1-T_{0}\right)\right)\right]\mathrm{d}X\right\} \nonumber \\
 & \overset{\left(a\right)}{\geq}\exp\left[-\lambda_{\textsf S^o}\int_{\mathbb{R}^{2}}c_{\textsf S^\prime}\left(1-T_{0}\right)\mathrm{d}X\right]\nonumber \\
 & \overset{\left(b\right)}{=}\exp\left\{ -\lambda_{\textsf S^o}c_{\textsf S^\prime}\int_{\mathbb{R}^{2}}\left[1-M_{2}(\int_{\mathbb{R}^{2}}v(J+Z)f_{\textsf S}(Z)\mathrm{d}Z)\right]\cdot\int_{\mathbb{R}^{2}}f_{\textsf S^\prime}(J-X)\mathrm{d}x\cdot\mathrm{d}J\right\} \nonumber \\
 & \overset{\left(c\right)}{=}G_{\Phi_{\textsf M}}(\lambda_{\textsf S}),\label{eq:Gsocp>Gmcp}
\end{flalign}
where $\left(a\right)$ follows from the fact that $1-\exp(-\theta x)\leq\theta x,\:\theta\geq0$,
$\left(b\right)$ comes from the change of variables $J=X+Y$ and interchanging integrals, $\left(c\right)$
follows from the fact that $\int_{\mathbb{R}^{2}}f_{\textsf S^\prime}(J-X)\mathrm{d}X=1$ and the expression of the PGF of MCP.

According to (\ref{eq:Bounds_MCP}), we have
\begin{equation}
G_{\Phi_{\textsf S}}(\lambda_{\textsf S}) \geq G_{\Phi_{\textsf M}}(\lambda_{\textsf S}) \geq G_{\Phi_{\textsf P}}(\lambda_{\textsf S}). \label{eq:Gsocp>Gppp}
\end{equation}

2) The upper bound of $G_{\Phi_{\textsf S}}(\lambda_{\textsf S})$:
\begin{flalign}
G_{\Phi_{\textsf S}}(\lambda_{\textsf S}) & =\exp\left\{ -\lambda_{\textsf S^o}\int_{\mathbb{R}^{2}}\left[1-\exp\left(-c_{\textsf S^\prime}\left(1-T_{0}\right)\right)\right]\mathrm{d}X\right\} \nonumber \\
 & \overset{\left(a\right)}{\leq}\exp\left[-\lambda_{\textsf S^o}\int_{\mathbb{R}^{2}}\frac{c_{\textsf S^\prime}\left(1-T_{0}\right)}{1+c_{\textsf S^\prime}\left(1-T_{0}\right)}\mathrm{d}X\right]\nonumber \\
 & \overset{\left(b\right)}{\leq}\exp\left[\frac{-\lambda_{\textsf S^o}c_{\textsf S^\prime}}{1+c_{\textsf S^\prime}}\int_{\mathbb{R}^{2}}\left(1-T_{0}\right)\mathrm{d}X\right]\nonumber \\
 & \overset{\left(c\right)}{=}\exp\left\{ \frac{-\lambda_{\textsf S^o}c_{\textsf S^\prime}}{1+c_{\textsf S^\prime}}\int_{\mathbb{R}^{2}}\int_{\mathbb{R}^{2}}\left[1-\exp(-c_{\textsf S}T^{\prime})\right]f_{\textsf S^\prime}(Y)\mathrm{d}Y\mathrm{d}X\right\} \nonumber \\
 & \overset{\left(d\right)}{\leq}\exp\left\{ \frac{-\lambda_{\textsf S^o}c_{\textsf S^\prime}c_{\textsf S}}{\left(1+c_{\textsf S^\prime}\right)\left(1+c_{\textsf S}\right)}\int_{\mathbb{R}^{2}}\int_{\mathbb{R}^{2}}\int_{\mathbb{R}^{2}}\left(1-v(X+Y+Z)\right)f_{\textsf S}(Z)\mathrm{d}Z\cdot f_{\textsf S^\prime}(Y)\mathrm{d}Y\mathrm{d}X\right\} \nonumber \\
 & \overset{\left(e\right)}{=}\exp\left[\frac{-\lambda_{\textsf S^o}c_{\textsf S^\prime}c_{\textsf S}}{\left(1+c_{\textsf S^\prime}\right)\left(1+c_{\textsf S}\right)}\int_{\mathbb{R}^{2}}\left(1-v(X)\right)\mathrm{d}X\right]\nonumber \\
 & =G_{\Phi_{\textsf P}}\left(\frac{\lambda_{\textsf S}}{(1+c_{\textsf S^\prime})(1+c_{\textsf S})}\right),
\end{flalign}
where $\left(a\right)$ comes from the fact that $\exp(-\theta x)\leq\left(1+\theta x\right)^{-1}$,
$\left(b\right)$ follows from the fact that $T_{0}=\int_{\mathbb{R}^{2}}M_{2}(\int_{\mathbb{R}^{2}}v(X+Y+Z)f_{\textsf S}(Z)\mathrm{d}Z)f_{\textsf S^\prime}(Y)\mathrm{d}Y\geq0$
(because $M_{2}(\int_{\mathbb{R}^{2}}v(X+Y+Z)f_{\textsf S}(Z)\mathrm{d}Z)>0$
and $f_{\textsf S^\prime}(Y)\geq0$), $\left(c\right)$ comes from $T^{\prime}=\int_{\mathbb{R}^{2}}(1-v(X+Y+Z))f_{\textsf S}(Z)\mathrm{d}Z$,
$\left(d\right)$ follows from the fact that $\exp(-\theta x)\leq\left(1+\theta x\right)^{-1}$
and $0\leq \theta \leq1$,
$\left(e\right)$ comes from the change of variables, interchanging integrals, and the fact that $\int_{\mathbb{R}^{2}}f_{\textsf S}(X)\mathrm{d}X=\int_{\mathbb{R}^{2}}f_{\textsf S^\prime}(X)\mathrm{d}X=1$.

3) The lower bound of $G_{\Phi_{\textsf S}^{!}}(\lambda_{\textsf S})$:

According to (\ref{eq:Coditional PGF_SOCP}), the conditional PGF of the SOCP is
\begin{flalign}
G_{\Phi_{\textsf S}^{!}}(\lambda_{\textsf S}) & =G_{\Phi_{\textsf S}}(\lambda_{\textsf S})\underbrace{M_{1}\left[\int_{\mathbb{R}^{2}}M_{2}(\int_{\mathbb{R}^{2}}v(X+Y+Z)f_{\textsf S}(Z)\mathrm{d}Z)f_{\textsf S^\prime}(Y)\mathrm{d}Y\right]}_{T_1}\nonumber \\
 & \cdot \underbrace{ \int_{\mathbb{R}^{2}}M_{2}(\int_{\mathbb{R}^{2}}v(X+Y+Z)f_{\textsf S}(Z)\mathrm{d}Z)f_{\textsf S}(Y)\mathrm{d}Y}_{T_2}. \label{eq:pf_SOCP}
\end{flalign}
Thus, the lower bound of $G_{\Phi_{\textsf S}^{!}}(\lambda_{\textsf S})$ is derived by bounding the following three terms, called $G_{\Phi_{\textsf S}}(\lambda_{\textsf S})$, $T_1$, and $T_2$.

First, the lower bound of $G_{\Phi_{\textsf S}}(\lambda_{\textsf S})$ is given in (\ref{eq:Gsocp>Gppp}).

Next, the lower bound of $T_1$ is calculated as follows:
\begin{flalign}
T_{1} & =\exp\left\{ -c_{\textsf S^\prime}\int_{\mathbb{R}^{2}}\left[1-\exp\left(-c_{\textsf S}\int_{\mathbb{R}^{2}}\left(1-v(X+Y+Z)\right)f_{\textsf S}(Z)\mathrm{d}Z\right)\right]f_{\textsf S^\prime}(Y)\mathrm{d}Y\right\} \nonumber \\
 & \overset{\left(a\right)}{\geq}\exp\left\{ -c_{\textsf S^\prime}c_{\textsf S}\int_{\mathbb{R}^{2}}(1-v(J))f_{\textsf S}(J-Y-X)f_{\textsf S^\prime}(Y)\mathrm{d}Y\mathrm{d}J\right\} \nonumber \\
 & \overset{\left(b\right)}{=}\exp\left[-c_{\textsf S^\prime}c_{\textsf S}\int_{\mathbb{R}^{2}}(1-v(J))f_{\textsf S}\star f_{\textsf S^\prime}(J-X)\mathrm{d}J\right]\nonumber \\
 & \overset{\left(c\right)}{\geq} exp\left[ -c_{\textsf S^\prime} c_{\textsf S} \widehat{f_{\textsf S}\star f_{\textsf S^\prime}} \int_{\mathbb{R}^{2}}(1-v(J)) \mathrm{d}J\right],\label{eq:T1}
\end{flalign}
where $\left(a\right)$ comes from the fact that $1-\exp(-\theta x)\leq\theta x,\:\theta\geq0$
and the change of variables $J=X+Y+Z$, $\left(b\right)$ follows from the definition of convolution $f_{\textsf S} \star f_{\textsf S^\prime}$, $\left(c\right)$ comes from $\widehat{f_{\textsf S}\star f_{\textsf S^\prime}} = \sup_{X \in \mathbb{R}^{2}}(f_{\textsf S}\star f_{\textsf S^\prime})(X)$.

Based on Young's inequality in \cite{YongEquality} ($\|f \star g\|_{r} \leq \|f\|_{p}\|g\|_{q}$, where $1/p+1/q=1/r+1$), we have $\widehat{f_{\textsf S}\star f_{\textsf S^\prime}} \leq \min \{\|f_{\textsf S^\prime}\|_{\infty} \|f_{\textsf S}\|_{1}, \|f_{\textsf S^\prime}\|_{1} \|f_{\textsf S}\|_{\infty}\}=\underbrace{\min\left\{ \frac{1-\exp\left(\frac{-R_{\textsf S^\prime}^{2}}{2\sigma^{2}}\right)}{\pi R_{\textsf S^\prime}^{2}+2\pi\sigma^{2}\left(\exp\left(\frac{-R_{\textsf S^\prime}^{2}}{2\sigma^{2}}\right)-1\right)},\frac{1}{\pi R_{\textsf S}^{2}}\right\}}_{\gamma}$.

Hence, we have
\begin{equation}
T_{1} \geq \exp\left[-c_{\textsf S^\prime}c_{\textsf S}\gamma\int_{\mathbb{R}^{2}}(1-v(J))\mathrm{d}J\right].
\end{equation}

Next, the lower bound of $T_2$ is given as:
\begin{flalign}
T_{2} & =\int_{\mathbb{R}^{2}}\exp\left[-c_{\textsf S}\int_{\mathbb{R}^{2}}(1-v(X+Y+Z))f_{\textsf S}(Z)\mathrm{d}Z\right]f_{\textsf S}(Y)\mathrm{d}Y\nonumber \\
 & \overset{\left(a\right)}{\geq}\exp\left[-c_{\textsf S}\int_{\mathbb{R}^{2}}\int_{\mathbb{R}^{2}}(1-v(X+Y+Z))f_{\textsf S}(Z)\mathrm{d}Z f_{\textsf S}(Y)\mathrm{d}Y\right]\nonumber \\
 & \overset{\left(b\right)}{=}\exp\left[-c_{\textsf S}\int_{\mathbb{R}^{2}}(1-v(J))f_{\textsf S}\star f_{\textsf S}(J-X)\mathrm{d}J\right]\nonumber \\
 & \overset{\left(c\right)}{\geq}\exp\left[-c_{\textsf S} \widehat{f_{\textsf S}\star f_{\textsf S}} \int_{\mathbb{R}^{2}}(1-v(J))\mathrm{d}J\right]\nonumber \\
 & \overset{\left(D\right)}{\geq}\exp\left[\frac{-c_{\textsf S}}{\pi R_{\textsf S}^{2}}\int_{\mathbb{R}^{2}}\left(1-v(J)\right)\mathrm{d}J\right],\label{eq:T2}
\end{flalign}
where $\left(a\right)$ comes from the fact that $f(x)=\exp(-x)$ is
convex and $\mathbb{E}\left[f(x)\right] \geq f(\mathbb{E}(x))$, $\left(b\right)$
follows from the change of variables $J=X+Y+Z$ and the definition of convolution $f_{\textsf S} \star f_{\textsf S}$, $\left(c\right)$ comes from $\widehat{f_{\textsf S}\star f_{\textsf S}} = \sup_{X \in \mathbb{R}^{2}}(f_{\textsf S}\star f_{\textsf S})(X)$, $\left(d\right)$
comes from the Young's inequality $\widehat{f_{\textsf S}\star f_{\textsf S}} \leq \|f_{\textsf S}\|_{\infty} \|f_{\textsf S}\|_{1}=\frac{1}{\pi R_{\textsf S}^{2}}$
.

Combining (\ref{eq:pf_SOCP}), (\ref{eq:Gsocp>Gppp}), (\ref{eq:T1}), and (\ref{eq:T2}), the lower bound of $G_{\Phi_{\textsf S}^{!}}$ is given as follows:
\begin{flalign}
G_{\Phi_{\textsf S}^{!}}(\lambda_{\textsf S}) &= G_{\Phi_{\textsf S}}(\lambda_{\textsf S}) \cdot T_1 \cdot T_2 \nonumber \\
& \geq \exp\!\left[\!-\lambda_{\textsf S} \int_{\mathbb{R}^{2}}(1\!-\!v(J))\mathrm{d}J\right] \exp\!\left[\!-c_{\textsf S^\prime}c_{\textsf S}\gamma\int_{\mathbb{R}^{2}}(1\!-\!v(J))\mathrm{d}J\right] \exp\!\left[\frac{-c_{\textsf S}}{\pi R_{\textsf S}^{2}}\int_{\mathbb{R}^{2}}\left(1\!-\!v(J)\right)\mathrm{d}J\right]\nonumber \\
& = \exp\left[-\left(\lambda_{\textsf S} + c_{\textsf S^\prime}c_{\textsf S}\gamma +  \frac{c_{\textsf S}}{\pi R_{\textsf S}^{2}} \right)\int_{\mathbb{R}^{2}}(1-v(J))\mathrm{d}J\right] = G_{\Phi_{\textsf P}^{!}}\left(\lambda_{\textsf S} + c_{\textsf S^\prime}c_{\textsf S}\gamma +  \frac{c_{\textsf S}}{\pi R_{\textsf S}^{2}} \right).
\end{flalign}

4) The upper bound of $G_{\Phi_{\textsf S}^{!}}(\lambda_{\textsf S})$:

\begin{flalign}
G_{\Phi_{\textsf S}^{!}}(\lambda_{\textsf S}) &= G_{\Phi_{\textsf S}}(\lambda_{\textsf S}) \cdot T_1 \cdot T_2 \overset{\left(a\right)} {\leq} G_{\Phi_{\textsf S}}(\lambda_{\textsf S}) \overset{\left(b\right)} {\leq} G_{\Phi_{\textsf P}^{!}}\left(\frac{\lambda_{\textsf S}}{\left(1+c_{\textsf S^\prime}\right)\left(1+c_{\textsf S}\right)}\right).
\end{flalign}
where (a) comes from $0 \leq T_1 \leq 1$ and $0 \leq T_2 \leq 1$ and (b) follows form Lemma \ref{lm:bounds_PGF_SOCP}.

\subsection{Proof of Lemma \ref{Lem:JointSuccProb:PPP}} \label{Pf:JSP_PPP}
According to (\ref{eq:JSP_MU_original}), $\mathcal{P}_{m \textsf P}^{(n)}(\lambda, r)$ is given as:
\begin{align}
&\mathcal{P}_{m \textsf P}^{(n)}(\lambda_s, r_m)\nonumber \\
& \overset{(a)}{=} \mathbb{E}_{X_m}^{!}\left[\prod_{X \in \Phi_m} \left(1+\frac{\beta_{m} \widetilde{g}(X,r_m)} {r_m^{-\alpha}}\right)^{-n}\right] \mathbb{E}\left[\prod_{X \in \Phi_{\textsf P}} \left(1+\frac{\beta_{m}P_{s}|X|^{-\alpha}}{P_{m} r_m^{-\alpha}}\right)^{-n}\right] \nonumber \\
& \overset{(b)}{=} \!\!\exp\left\{ \!\!- 2 \pi \lambda_m \int_{r}^{\infty} \left[1\!-\!\left(1\!+\! \frac{\beta_{m} r^{-\alpha}} {r_m^{-\alpha}}\right)^{-n}\right] \!r \mathrm{d} r\!\right\} \exp \!\left\{\!- 2 \pi \lambda_s \!\!\int_{0}^{\infty} \!\!\left[1\!-\!\left(1\!+\!\frac{\beta_{m}P_{s}r^{-\alpha}}{P_{m} r_m^{-\alpha}}\right)^{-n}\right]\! r \mathrm{d}r \right\} \nonumber \\
& \overset{(c)}{=} \exp\left[-\lambda_{m}Q_{n}(\beta_{m}) r_m^{2}\right] \exp\left[- \lambda_s \left(\frac{\beta_{m}P_{s}}{P_{m}}\right)^{\delta}U_{n} r_m^{2}\right].
\end{align}
where (a) comes from (\ref{eq:JSP_MU_original}), (b) follows from (\ref{eq:PGF_PPP}) and  converting from Cartesian to polar coordinates, (c) comes from (22) in \cite{InterfCorreglobecom} with $K=1$ and Theorem 1 in \cite{IntefCorrDiversityPolynomials} with $\theta = \frac{\beta_m P_s}{P_m}$ and $p=1$, the function $Q_{n}(\beta_{m})$ and constant $U_{n}$ are defined in (\ref{eq:Qn}) and (\ref{eq:Un}), respectively.\\

Following the similar steps, $\mathcal{P}_{s \textsf P}^{(n)}(\lambda, r)$ is derived as shown in Lemma \ref{Lem:JointSuccProb:PPP}.

\subsection{Proof of Theorem \ref{Thm:Bounds_JSP}} \label{Pf:Bounds_JSP}
Based on the expressions of the conditional joint success probability (Lemma \ref{lm:ConJointSuccProb}) and the bounds of PGF and the conditional PGF for MCP and SOCP (from (\ref{eq:Bounds_MCP}) to (\ref{eq:Bounds_Conditional_SOCP})), the bounds of the conditional joint success probabilities for the MCP model and SOCP model are bounded by their counterparts for the PPP model as follows:
\begin{equation}
\mathcal{P}_{m \textsf P}^{(n)}(\lambda_{\textsf M},r_m)
\leq \mathcal{P}_{m \textsf M}^{(n)}(r_m)
\leq \mathcal{P}_{m \textsf P}^{(n)}(\frac{\lambda_{\textsf M}}{1+c_{\textsf M}},r_m)
\end{equation}
\begin{equation}
\mathcal{P}_{s \textsf P}^{(n)}(\lambda_{\textsf M}+\frac{c_{\textsf M}}{\pi R_{\textsf M}^2},r_s)
\leq \mathcal{P}_{s \textsf M}^{(n)}(r_s)
\leq \mathcal{P}_{s \textsf P}^{(n)}(\frac{\lambda_{\textsf M}}{1+c_{\textsf M}},r_s)
\end{equation}
\begin{equation}
\mathcal{P}_{m \textsf P}^{(n)}(\lambda_{\textsf S},r_m)
\leq \mathcal{P}_{m \textsf M}^{(n)}(\lambda_{\textsf S},r_m)
\leq \mathcal{P}_{m \textsf S}^{(n)}(r_m)
\leq \mathcal{P}_{m \textsf P}^{(n)}(\frac{\lambda_{\textsf S}}{(1+c_{\textsf S^\prime})(1+c_{\textsf S})},r_m)
\end{equation}
\begin{equation}
\mathcal{P}_{s \textsf P}^{(n)}(\lambda_{\textsf S}+ c_{\textsf S^\prime} c_{\textsf S} + \frac{c_{\textsf S}}{\pi R_{\textsf S}^2},r_s)
\leq \mathcal{P}_{s \textsf S}^{(n)}(r_s)
\leq \mathcal{P}_{s \textsf P}^{(n)}(\frac{\lambda_{\textsf S}}{(1+c_{\textsf S^\prime})(1+c_{\textsf S})},r_s),
\end{equation}
where $\mathcal{P}_{m \textsf P}^{(n)}(\lambda, r)$ and $\mathcal{P}_{s \textsf P}^{(n)}(\lambda, r)$ denote the conditional joint success probability (given the serving distance $r$) for the typical MU and SU in the PPP model.

By taking expectation with respect to the distance distribution in (\ref{eq:pdf_r}), the bounds of the joint success probabilities are derived as shown in Theorem \ref{Thm:Bounds_JSP}.

\bibliographystyle{IEEEtran}
%\bibliography{9D__Papers_ref}

\begin{thebibliography}{10}
\providecommand{\url}[1]{#1}
\csname url@samestyle\endcsname
\providecommand{\newblock}{\relax}
\providecommand{\bibinfo}[2]{#2}
\providecommand{\BIBentrySTDinterwordspacing}{\spaceskip=0pt\relax}
\providecommand{\BIBentryALTinterwordstretchfactor}{4}
\providecommand{\BIBentryALTinterwordspacing}{\spaceskip=\fontdimen2\font plus
\BIBentryALTinterwordstretchfactor\fontdimen3\font minus
  \fontdimen4\font\relax}
\providecommand{\BIBforeignlanguage}[2]{{%
\expandafter\ifx\csname l@#1\endcsname\relax
\typeout{** WARNING: IEEEtran.bst: No hyphenation pattern has been}%
\typeout{** loaded for the language `#1'. Using the pattern for}%
\typeout{** the default language instead.}%
\else
\language=\csname l@#1\endcsname
\fi
#2}}
\providecommand{\BIBdecl}{\relax}
\BIBdecl

\bibitem{HCN-K-Tier}
H.~Dhillon, R.~Ganti, F.~Baccelli, and J.~Andrews, ``Modeling and analysis of
  {K}-tier downlink heterogeneous cellular networks,'' \emph{IEEE J. Sel. Areas
  Commun.}, vol.~30, no.~3, pp. 550--560, Apr. 2012.

\bibitem{HCN-Model}
Y.~Lin, W.~Bao, W.~Yu, and B.~Liang, ``Optimizing user association and spectrum
  allocation in {HetNets}: A utility perspective,'' \emph{IEEE J. Sel. Areas
  Commun.}, vol.~33, no.~6, pp. 1025--1039, Jun. 2015.

\bibitem{Stochastic-Geometry-HCN-Tutorial}
H.~ElSawy, E.~Hossain, and M.~Haenggi, ``Stochastic geometry for modeling,
  analysis, and design of multi-tier and cognitive cellular wireless networks:
  A survey,'' \emph{IEEE Comm. Surveys and Tutorials}, vol.~15, no.~3, pp.
  996--1019, Jun. 2013.

\bibitem{StoGeo}
M.~Haenggi, J.~Andrews, F.~Baccelli, O.~Dousse, and M.~Franceschetti,
  ``Stochastic geometry and random graphs for the analysis and design of
  wireless networks,'' \emph{IEEE J. Sel. Areas Commun.}, vol.~27, no.~7, pp.
  1029--1046, Sep. 2009.

\bibitem{FlexibleCellAssociation}
H.-S. Jo, Y.~J. Sang, P.~Xia, and J.~Andrews, ``Heterogeneous cellular networks
  with flexible cell association: A comprehensive downlink {SINR} analysis,''
  \emph{IEEE Trans. Wireless Commun.}, vol.~11, no.~10, pp. 3484--3495, Oct.
  2012.

\bibitem{StructuredSpectrumAllocation}
W.~Bao and B.~Liang, ``Structured spectrum allocation and user association in
  heterogeneous cellular networks,'' in \emph{Proc. IEEE INFOCOM}, Toronto, ON,
  CA, Arp. 2014, pp. 1069--1077.

\bibitem{HCN-FFR}
T.~Novlan, R.~Ganti, A.~Ghosh, and J.~Andrews, ``Analytical evaluation of
  fractional frequency reuse for ofdma cellular networks,'' \emph{IEEE Trans.
  Wireless Commun.}, vol.~10, no.~12, pp. 4294--4305, Dec. 2011.

\bibitem{Tony_HCN_Offloading}
P.~S. Yu, J.~Lee, T.~Q.~S. Quek, and Y.~W.~P. Hong, ``Traffic offloading in
  heterogeneous networks with energy harvesting personal cells-network
  throughput and energy efficiency,'' \emph{IEEE Trans. Wireless Commun.},
  vol.~15, no.~2, pp. 1146--1161, Feb 2016.

\bibitem{HongGuang_15TWC}
H.~Sun, M.~Wildemeersch, M.~Sheng, and T.~Q.~S. Quek, ``{D2D} enhanced
  heterogeneous cellular networks with dynamic {TDD},'' \emph{IEEE Trans.
  Wireless Commun.}, vol.~14, no.~8, pp. 4204--4218, Aug. 2015.

\bibitem{Tony_D2D_HCN}
H.~H. Yang, J.~Lee, and T.~Q.~S. Quek, ``Heterogeneous cellular network with
  energy harvesting-based {D2D} communication,'' \emph{IEEE Trans. Wireless
  Commun.}, vol.~15, no.~2, pp. 1406--1419, Feb. 2016.

\bibitem{SpecSharing_Cellular_Adhoc}
K.~Huang, V.~K.~N. Lau, and Y.~Chen, ``Spectrum sharing between cellular and
  mobile ad hoc networks: transmission-capacity trade-off,'' \emph{IEEE J. Sel.
  Areas Commun.}, vol.~27, no.~7, pp. 1256--1267, Sep. 2009.

\bibitem{Tony_EE_HCN}
Y.~S. Soh, T.~Q.~S. Quek, M.~Kountouris, and H.~Shin, ``Energy efficient
  heterogeneous cellular networks,'' \emph{IEEE J. Sel. Areas Commun.},
  vol.~31, no.~5, pp. 840--850, May 2013.

\bibitem{HCN-SpatiotemporalCooperation}
G.~Nigam, P.~Minero, and M.~Haenggi, ``Spatiotemporal cooperation in
  heterogeneous cellular networks,'' \emph{IEEE J. Sel. Areas Commun.},
  vol.~33, no.~6, pp. 1253--1265, Jun. 2015.

\bibitem{InterfCorrela-ICIC}
X.~Zhang and M.~Haenggi, ``A stochastic geometry analysis of inter-cell
  interference coordination and intra-cell diversity,'' \emph{IEEE Trans.
  Wireless Commun.}, vol.~13, no.~12, pp. 6655--6669, Dec. 2014.

\bibitem{MulticellCoop_KB}
K.~Huang and J.~G. Andrews, ``An analytical framework for multicell cooperation
  via stochastic geometry and large deviations,'' \emph{IEEE Trans. Inf.
  Theory}, vol.~59, no.~4, pp. 2501--2516, Apr. 2013.

\bibitem{HCN_ApproxSIRAnalysis}
H.~Wei, N.~Deng, W.~Zhou, and M.~Haenggi, ``Approximate {SIR} analysis in
  general heterogeneous cellular networks,'' \emph{IEEE Trans. Commun.},
  vol.~64, no.~3, pp. 1259--1273, Mar. 2016.

\bibitem{InterferCor_NonPoisson_GC}
\BIBentryALTinterwordspacing
J.~Wen, M.~Sheng, K.~Huang, and J.~Li, ``Analysis of interference correlation
  in non-poisson networks,'' in \emph{IEEE GLOBECOM}, 2016. [Online].
  Available: \url{https://arxiv.org/pdf/1604.04166v1.pdf}
\BIBentrySTDinterwordspacing

\bibitem{HCN-Dependence}
N.~Deng, W.~Zhou, and M.~Haenggi, ``Heterogeneous cellular network models with
  dependence,'' \emph{IEEE J. Sel. Areas Commun.}, vol.~33, no.~10, pp.
  2167--2181, Oct. 2015.

\bibitem{DPP1}
Y.~Li, F.~Baccelli, H.~Dhillon, and J.~Andrews, ``Statistical modeling and
  probabilistic analysis of cellular networks with determinantal point
  processes,'' \emph{IEEE Trans. Commun.}, vol.~63, no.~9, pp. 3405--3422, Sep.
  2015.

\bibitem{DPP2}
R.~Ganti, F.~Baccelli, and J.~Andrews, ``Series expansion for interference in
  wireless networks,'' \emph{IEEE Trans. Inf. Theory}, vol.~58, no.~4, pp.
  2194--2205, Apr. 2012.

\bibitem{GPP1}
N.~Deng, W.~Zhou, and M.~Haenggi, ``The {Ginibre} point process as a model for
  wireless networks with repulsion,'' \emph{IEEE Trans. Wireless Commun.},
  vol.~14, no.~1, pp. 107--121, Jan. 2015.

\bibitem{HCN_PoissonCluster}
J.~Young, M.~Hasna, and A.~Ghrayeb, ``Modeling heterogeneous cellular networks
  interference using poisson cluster processes,'' \emph{IEEE J. Sel. Areas
  Commun.}, vol.~33, no.~10, pp. 2182--2195, Oct. 2015.

\bibitem{PoissonClusterProcess_TIT}
R.~Ganti and M.~Haenggi, ``Interference and outage in clustered wireless ad hoc
  networks,'' \emph{IEEE Trans. Inf. Theory}, vol.~55, no.~9, pp. 4067--4086,
  Sep. 2009.

\bibitem{HCN_SOCP}
S.~Asif and K.~Kwak, ``Downlink coverage and rate analysis of two-tier
  networks,'' \emph{IEEE Wireless Comms. Lett.}, vol.~4, no.~2, pp. 133--136,
  Apr. 2015.

\bibitem{InterfCorreThreeSources}
U.~Schilcher, C.~Bettstetter, and G.~Brandner, ``Temporal correlation of
  interference in wireless networks with {Rayleigh} block fading,'' \emph{IEEE
  Trans. Mobile Comp.}, vol.~11, no.~12, pp. 2109--2120, Dec. 2012.

\bibitem{InterferenceCorreLetter}
R.~Ganti and M.~Haenggi, ``Spatial and temporal correlation of the interference
  in {ALOHA} ad hoc networks,'' \emph{IEEE Commun. Lett.}, vol.~13, no.~9, pp.
  631--633, Sep. 2009.

\bibitem{IntefCorrDiversityPolynomials}
M.~Haenggi and R.~Smarandache, ``Diversity polynomials for the analysis of
  temporal correlations in wireless networks,'' \emph{IEEE Trans. Wireless
  Commun.}, vol.~12, no.~11, pp. 5940--5951, Nov. 2013.

\bibitem{InterfCorreDiversityLoss}
M.~Haenggi, ``Diversity loss due to interference correlation,'' \emph{IEEE
  Commun. Lett.}, vol.~16, no.~10, pp. 1600--1603, Oct. 2012.

\bibitem{CorrelationMobileRandomNet}
Z.~Gong and M.~Haenggi, ``Interference and outage in mobile random networks:
  Expectation, distribution, and correlation,'' \emph{IEEE Trans. Mobile
  Comp.}, vol.~13, no.~2, pp. 337--349, Feb. 2014.

\bibitem{LocalDelay_MAC}
Y.~Zhong, W.~Zhang, and M.~Haenggi, ``Managing interference correlation through
  random medium access,'' \emph{IEEE Trans. Wireless Commun.}, vol.~13, no.~2,
  pp. 928--941, Feb. 2014.

\bibitem{HARQ-CorrInterf-adhoc}
H.~Ding, S.~Ma, C.~Xing, Z.~Fei, Y.~Zhou, and C.~Chen, ``Analysis of hybrid
  {ARQ} in ad hoc networks with correlated interference and feedback errors,''
  \emph{IEEE Trans. Wireless Commun.}, vol.~12, no.~8, pp. 3942--3955, Aug.
  2013.

\bibitem{zhong2016delay}
Y.~Zhong, T.~Q.~S. Quek, and X.~Ge, ``{Heterogeneous Cellular Networks with
  Spatio-Temporal Traffic: Delay Analysis and Scheduling},'' \emph{arXiv
  preprint arXiv:1611.08067}, 2016.

\bibitem{zhong16Stability}
Y.~Zhong, M.~Haenggi, T.~Q.~S. Quek, and W.~Zhang, ``On the stability of static
  poisson networks under random access,'' \emph{IEEE Transactions on
  Communications}, vol.~64, no.~7, pp. 2985--2998, July 2016.

\bibitem{LocalDelay_PoissonNet}
M.~Haenggi, ``The local delay in poisson networks,'' \emph{IEEE Trans. Inf.
  Theory}, vol.~59, no.~3, pp. 1788--1802, Mar. 2013.

\bibitem{StoGeoBook-Martin}
------, \emph{Stochastic Geometry for Wireless Networks}.\hskip 1em plus 0.5em
  minus 0.4em\relax Cambridge University Press, 2012.

\bibitem{IntegralTable}
I.~S. Gradshteyn and I.~M. Ryzhik, \emph{Table of Integrals, Series, and
  Products}, 7th~ed.\hskip 1em plus 0.5em minus 0.4em\relax Academic Press,
  2007.

\bibitem{YongEquality}
G.~B. Folland, \emph{Real Analysis, Modern Techniques and Their Applications},
  2nd~ed.\hskip 1em plus 0.5em minus 0.4em\relax New Yourk: Wiley, 1999.

\bibitem{InterfCorreglobecom}
J.~Wen, M.~Sheng, B.~Liang, X.~Wang, Y.~Zhang, and J.~Li, ``Correlations of
  interference and link successes in heterogeneous cellular networks,'' in
  \emph{IEEE GLOBECOM}, San Diego, CA, Dec. 2015, pp. 1--6.

\end{thebibliography}

\end{document}